\documentclass[11pt,letterpaper]{article}
\usepackage[utf8]{inputenc}
\usepackage{fullpage}
\usepackage{amsthm}
\usepackage{xcolor}
\usepackage{mathtools,amsmath,amssymb}
\usepackage{mathrsfs}
\usepackage{algorithm, algpseudocode}
\usepackage{caption}
\usepackage{subcaption}
\usepackage{comment}
\usepackage{authblk}
\usepackage[colorlinks=true, allcolors=blue]{hyperref}
\usepackage[capitalise,nameinlink]{cleveref}
\usepackage[style=alphabetic, backend=biber, minalphanames=3, maxalphanames=5, maxbibnames=99, sorting=nyt]{biblatex}

\usepackage{pifont,mdframed}
\usepackage{tikz}
\usepackage{quantikz} 
\usepackage{amsmath}
\usetikzlibrary{decorations.pathreplacing,calligraphy}

\usepackage{multirow}

\usepackage{bm}
\usepackage{enumitem}
\setlist[itemize]{topsep=3pt}
\setlist[enumerate]{topsep=3pt}
\usepackage{mathtools}
\usepackage{complexity}
\newlang{\ApxSim}{ApxSim}
\newlang{\LHPlusE}{LHPlusE}
\newlang{\LH}{LH}
\newlang{\kSSH}{kSSH}
\usepackage{dsfont}
\usepackage{float}
\usepackage{thm-restate}
\usepackage{bbm}
\usepackage{amsthm}
\usepackage{thmtools,thm-restate}

\usepackage[colorlinks=true, allcolors=blue]{hyperref}
\usepackage[nameinlink,capitalise]{cleveref}
\hypersetup{
    citecolor={violet}
}

\theoremstyle{plain}
\newtheorem{theorem}{Theorem}
\newtheorem{corollary}[theorem]{Corollary}

\newtheorem{lemma}[theorem]{Lemma}
\newtheorem{claim}[theorem]{Claim}
\newtheorem{fact}[theorem]{Fact}

\newtheorem{definition}[theorem]{Definition}
\newtheorem{problem}[theorem]{Problem}

\newtheorem{remark}[theorem]{Remark}
\theoremstyle{plain}

\DeclarePairedDelimiterXPP{\bigo}[1]{O}{(}{)}{}{#1}
\DeclarePairedDelimiterXPP{\bigomega}[1]{$\Omega$}{(}{)}{}{#1}

\newcommand{\defeq}{\stackrel{\mathrm{\scriptscriptstyle def}}{=}}

\renewcommand{\poly}{\textnormal{poly}}

\usepackage[margin=1.in]{geometry}
\usepackage{physics, graphicx}
\usepackage[compact]{titlesec}
\titlespacing{\subsection}{0pt}{1.5ex}{0ex}
\titlespacing{\subsubsection}{0pt}{1.5ex}{0ex}
\titlespacing{\subsubsection}{0pt}{1ex}{0ex}
\titlespacing{\paragraph}{0pt}{1.5ex}{1ex}
\usepackage{svg}
\usepackage{tikz}
\usetikzlibrary{quantikz2}
\usetikzlibrary{calc,positioning,arrows.meta}%,backgrounds,fit,decorations.pathreplacing,quantikz}

\usepackage{diagbox}
\usetikzlibrary{arrows.meta,patterns,decorations.pathreplacing}

\title{On the Complexity of Decoded Quantum Interferometry}
\author[1]{Kunal Marwaha\thanks{\href{mailto:kmarw@uchicago.edu}{kmarw@uchicago.edu}}}
 \author[1]{Bill Fefferman}
 \author[2]{Alexandru Gheorghiu}
 \author[2]{Vojtech Havlicek}
\affil[1]{University of Chicago}
\affil[2]{IBM Quantum}
\date{\vspace{-5.5em}}
\begin{document}
\maketitle
\begin{abstract}
We study the complexity of Decoded Quantum Interferometry (DQI), a
quantum algorithm for approximate optimization~\cite{dqi}. First, we show that the algorithm resists classical simulation strategies based
on locating outputs with large probabilities. 
We then prove that DQI can be simulated at a low level of the
polynomial hierarchy, posing challenges to standard
quantum supremacy arguments. 
We further show that DQI is a constructive solution to a classical coding-theoretic bound
based on the MacWilliams identity. 
Lastly, we interpret DQI as preparing low-energy states of a quantum simple harmonic oscillator, a viewpoint we believe suggests a physics-motivated route to generalizing DQI. 
\end{abstract}
\vspace{1em}

\section{Introduction}

\subsection{Background}

Exponential quantum speedups are the primary motivation for building
quantum computers, yet examples are rare.
Prominent instances include Shor's algorithm for factoring and
computing discrete logarithms~\cite{shor1994algorithms}, sampling from
certain families of random quantum circuits
\cite{aaronson2011computational,bremner2011classical,
aaronson2016complexity,bouland2019complexity,morimae2019fine},
and cryptographic protocols based on interactive tests of quantumness
\cite{brakerski2021cryptographic,kahanamoku2022classically}.
In each case, the classical hardness of the associated computational
task is justified either by established cryptographic assumptions
(e.g.\ the hardness underlying RSA) or by widely believed
complexity-theoretic assumptions such as the non-collapse of the
polynomial hierarchy.

It is natural to ask where one might find
additional examples of
exponential quantum speedups. One promising direction originates from
\emph{Regev's reduction} or the \emph{quantum decoding problem}.
The idea is to use a decoder for a linear code $C$ in order to
find small codewords in its dual code $C^\perp$
\cite{aharonov2003latticeproblemquantumnp,regev2009lattices,
Chailloux:2024dll}.
Regev used this to provide evidence that the Learning With Errors (LWE) problem, which underlies much of post-quantum cryptography, resists quantum attacks. This is because an efficient solver for LWE (i.e. solving a decoding problem), either classical or quantum, can be turned into an efficient quantum algorithm for the Shortest Vector Problem (SVP). The latter is known to be \textsf{NP}-hard for certain parameter choices, making it unlikely that an efficient (quantum) algorithm exists.
More recently, Yamakawa and Zhandry used the same quantum reduction idea to construct a random oracle problem exhibiting an exponential quantum speedup
\cite{yamakawa2024verifiable}.
This has since inspired several candidate protocols for
verifiable quantum advantage
\cite{Chailloux:2024dll,quantumfire,sahai_yz}.

\emph{Decoded Quantum Interferometry} (DQI) is one of the algorithms inspired by this idea
\cite{dqi}.
It finds an approximate solution to a combinatorial optimization problem called
\emph{Max-LINSAT}: given a matrix
$B\in\mathbb{F}_p^{m\times n}$ over a finite field $\mathbb{F}_p$ and
sets $F_1,\dots,F_m\subseteq\mathbb{F}_p$,  find a vector $x\in\mathbb{F}_p^n$ satisfying as many
constraints of the form $\sum_j B_{ij}x_j\in F_i$.
The focus here is on the overdetermined regime $m>n$, where the system contains
more constraints than variables 
\cite{trevisan2014inapproximability,kramer2026tight}.

DQI approaches this problem as follows:
The matrix $B$ is interpreted as the generator matrix of a linear code
$C$, so that $B^T$ is the parity-check matrix of the dual code
$C^\perp$.
A key observation in~\cite{dqi} is that the expected number of
constraints satisfied by the DQI output is related to the decoding
radius of the dual code, a relationship referred to as the
\emph{semicircle law}.
Informally, the more errors that can be decoded in the dual code,
the larger the fraction of constraints that can be satisfied for the
primal problem. This means that DQI produces solutions whose expected number of satisfied
constraints depends on the performance of an efficient decoder for the
dual code. In \cite{dqi}, the authors present evidence that DQI achieves a higher satisfying fraction than known efficient classical algorithms. DQI works in a regime where {Max-LINSAT} is not known to be $\mathsf{NP}$-hard but also not obviously efficiently classically solvable.  For certain instances of Max-LINSAT that can be solved efficiently by DQI, in particular the variant known as
\emph{Optimal Polynomial Intersection} (OPI), the authors of
\cite{dqi} conjecture that no polynomial-time classical algorithm
matches the performance achieved by DQI. 

\subsection{Results} 
There has been much activity in understanding how DQI and its variants lead to quantum advantage~\cite{ct_softdecoders, bu2025decodedquantuminterferometrynoise, schmidhuber2025hamiltonian, anschuetz2025decoded, parekh2025no, kothari2025no, khattar2025verifiable}.  However, it remains open whether DQI admits a rigorous separation from
classical algorithms. We contribute to this research thread by clarifying the limits of classical simulation strategies for DQI, isolating the possible source of quantum speedup, and providing new interpretations of the
algorithm from coding theory and physics:

\begin{enumerate}
\item 
\emph{DQI resists classical simulation algorithms based on sparsity.} 

In Ref.~\cite{svdn13}, Schwarz and van den Nest showed that if the output
distribution of certain quantum algorithms contain 
 $1/\poly(n)$-large probabilities (peaks), then the
peaks can be found using an efficient randomized algorithm
\cite{goldreich1989hard,kushilevitz1991learning}.
 We show that the same simulation strategy works for DQI, i.e. any output probability that is at least
$1/\poly(n)$ can be found in classical polynomial-time. However, we then prove that no such heavy outputs typically exist in the output distributions of DQI.

\item
\emph{Obstructions to DQI-based quantum supremacy.} 

Several classes of quantum circuits, such as \textsf{IQP} or random circuits, are believed to be
classically hard to sample from because the existence of an efficient classical sampler would have unlikely computational complexity consequences, such as a collapse of the polynomial hierarchy
\cite{bremner2011classical,morimae2019fine, aaronson2011computational}. Here we identify an obstruction to using this argument for hardness of sampling from DQI distributions: the output probabilities can be computed by a polynomial-time classical algorithm. 
In fact, for any Max-LINSAT instance, we can sample from the output distribution of DQI up to multiplicative error
using
a polynomial-time classical algorithm with an \textsf{NP} oracle, i.e. $\BPP^\NP$. 
Note that this is not the same as solving the Max-LINSAT instance outright using the $\NP$ oracle. Instead, our algorithm can sample the output distribution of DQI.

\item
\emph{DQI implements a coding-theoretic bound with no known efficient classical algorithm.}

We show that DQI over $\mathbb{F}_2$ realizes a classical result
from coding theory relating the covering radius of a linear code $C$
to the decoding threshold of its dual code $C^\perp$.
The existence of solutions achieving the same performance as DQI appeared in
the coding theory literature as early as 1990~\cite{tietavainen}. 
However, the associated proof is not constructive.
We show that this proof can be interpreted as a polynomial-size
\emph{quantum circuit} equivalent to DQI.

\item
\emph{DQI prepares a truncated quantum harmonic oscillator state.}

We express DQI over $\mathbb{F}_2$ as a quantum harmonic oscillator embedded in $\mathbb{F}_2^m$.
The embedding restricts the oscillator to eigenstates below an energy
threshold determined by the decoding radius. We show that the optimal DQI state is then localized near the largest position allowed by
the energy threshold. 
We suggest a method to generalize DQI by adjusting the parameters of the quantum oscillator.

\end{enumerate}

\subsection{Organization}
In \Cref{sec:preliminaries} we review the coding-theoretic and
algorithmic background needed for DQI.
In \Cref{sec:peakfinding} we analyze sparsity-based classical simulation
strategies.
In \Cref{sec:stockmeyer} we place DQI sampling in the polynomial
hierarchy.
In \Cref{sec:semicirclelaw} we relate DQI to a
coding-theoretic bound, and in \Cref{sec:harmonic-oscillator} we interpret the
DQI state as a truncated quantum harmonic oscillator. We conclude by discussing our results in \Cref{sec:discussion}.

%%\pagebreak

\section{Preliminaries}
\label{sec:preliminaries}

\subsection{Coding theory}

We briefly recall the coding-theoretic terminology used in this paper;
see, e.g., \cite{codingtheorynotes_venkat,basu2023codingtheory} for background. Let $\mathbb{F}$ be a finite field.
A \emph{linear code} $C \subseteq \mathbb{F}^n$ of dimension $k$ is a
$k$-dimensional linear subspace of $\mathbb{F}^n$.
Equivalently, $C$ can be written as
\[
C = \{Gx : x \in \mathbb{F}^k\},
\]
where $G \in \mathbb{F}^{n\times k}$ is called a \emph{generator matrix}.
Such a code is referred to as an $[n,k]$ code. 
The \emph{distance} of a linear code is the minimum Hamming weight of a
nonzero codeword.
An $[n,k,d]$ code is a code of blocklength $n$, dimension $k$, and
distance $d$.
Every linear code $C \subseteq \mathbb{F}^n$ admits a
\emph{parity-check matrix}
$H \in \mathbb{F}^{(n-k)\times n}$ such that
\[
C = \{c \in \mathbb{F}^n : Hc = 0\}.
\]
Given $y \in \mathbb{F}^n$, the vector $Hy$ is called the
\emph{syndrome} of $y$. The \emph{dual code} of $C$ is defined as
\[
C^\perp := \{z \in \mathbb{F}^n : z \cdot c = 0
\;\; \forall c \in C\}.
\]
If $H$ is a parity-check matrix for $C$, then $H^T$
generates $C^\perp$.

\subsection{Decoded Quantum Interferometry}
\label{sub:setup}

Decoded Quantum Interferometry (DQI), introduced in
\cite{dqi}, is a quantum algorithm for approximately solving overdetermined linear constraint satisfaction problems. It solves the following two search optimization problems, Max-XORSAT and its generalization, Max-LINSAT.

\begin{problem}[Max-XORSAT search]
\label{problem:xorsatsearch}
Given $B \in \mathbb{F}_2^{m\times n}$ and $v \in \mathbb{F}_2^m$,
find $x \in \mathbb{F}_2^n$ maximizing the number of satisfied equations
in
\[
Bx = v \pmod 2 .
\]
Equivalently, maximize the objective function:
\begin{align}
\label{eqn:def_cost}
f(x) \defeq
\sum_{i=1}^m (-1)^{\sum_{j=1}^n B_{ij}x_j + v_i}\,.
\end{align}
\end{problem}
DQI  \cite{dqi} attacks Problem~\ref{problem:xorsatsearch} as follows: observe that $f(x)$ has Fourier sparsity at most $m$.
Writing $b_i \in \{0,1\}^n$ for the $i$-th row of $B$, we have
\[
f(x) = \sum_{i=1}^m (-1)^{v_i} (-1)^{b_i \cdot x}
= \sum_{i=1}^m (-1)^{v_i} \chi_{b_i}(x),
\]
where $\chi_a(x) = (-1)^{a \cdot x}$ is a character of $\mathbb{Z}_2^n$.
Thus $f$ is supported on at most $m$ Fourier characters. 
This enables efficient preparation of the state
\[
\frac{1}{\sqrt{m}} \sum_{i=1}^m (-1)^{v_i} \ket{b_i},
\]
whose Hadamard transform is
\begin{align}
H^{\otimes n}
\frac{1}{\sqrt{m}} \sum_{i=1}^m (-1)^{v_i} \ket{b_i}
&=
\frac{1}{\sqrt{2^n m}}
\sum_{x \in \{0,1\}^n}
\sum_{i=1}^m
(-1)^{b_i \cdot x + v_i}
\ket{x} =
\frac{1}{\sqrt{2^n m}}
\sum_{x \in \{0,1\}^n}
f(x)\ket{x}.
\end{align}

Sampling this state in the computational basis yields a sample $x$ with probability proportional to $|f(x)|^2$; this increases the chance of obtaining a large objective value compared to uniformly random choice of $x$. We remark that the Fourier-sparsity of Eq.~\ref{eqn:def_cost} was a key motivation for this work; and especially studying simulation of DQI using the Kushilevitz-Mansour algorithm in Sec.~\ref{sec:peakfinding}.

Following the exposition in \cite{dqi}, DQI amplifies this signal by preparing a state: \begin{align}
\label{eqn:dqi_state}
\ket{\psi_P}
\propto
\sum_{x \in \{0,1\}^n} P(f(x)) \ket{x},
\end{align} where $P$ is an appropriately normalized degree-$\ell$ polynomial. This proceeds as follows: 

\begin{enumerate}
    \item Prepare a weighted superposition of Dicke states: $ \ket{\psi} = \sum_k w_k \ket{D_{m,k}}$, where: \begin{align} \ket{D_{m,k}} = \sum_{y \in \mathbb{F}_2^m, |y|=k} \frac{1}{\sqrt{\binom{m}{k}}}\ket{y} \end{align}
    We assume that the weights $w_k$ are known and explicitly computable.
    \item Phase-encode $v$ by applying $Z_1^{v_1} \otimes \ldots \otimes Z_m^{v_m}$
    \item Given the $m \times n$ matrix $B$, compute $B^Ty$ into an ancilla register to get:
    \begin{align}
       \sum_{k=0}^\ell w_k \sum_{y \in \mathbb{F}_2^m, |y|=k} \frac{(-1)^{v \cdot y}}{\sqrt{{m \choose k}}} \ket{y} \ket{B^T y}
    \end{align}
    \item Uncompute $y$ to get a state:
    \begin{align}
        \sum_{k=0}^\ell w_k \sum_{y \in \mathbb{F}_2^m, |y|=k} \frac{(-1)^{v \cdot y}}{\sqrt{{m \choose k}}} \ket{B^T y}
    \end{align}
    This is enabled by interpreting $B^T$ as a generator matrix. Let
\[
C^\perp = \{ d \in \mathbb{F}_2^m : B^T d = 0\}
\]
be the dual code associated with $B$.
Suppose we are given a decoder for $C^\perp$ capable of correcting up to $\ell$ errors. Starting from a superposition of the form
\[
\sum_{y:\,|y|\le \ell} \alpha_y \ket{y}\ket{B^T y},
\]
we use the decoder coherently to recover $y$ from $B^T y$.
More precisely, the decoder implements a reversible map that,
on inputs of the form $B^T y$ with $|y|\le \ell$, outputs $y$.
This allows us to apply the transformation
\[
\ket{y}\ket{B^T y} \;\longmapsto\; \ket{0}\ket{B^T y},
\]
thereby uncomputing the $y$ register. This is possible because for $|y|\le \ell$, the vector $y$
is the unique low-weight solution to $B^T y$, allowing the
decoder to invert the map $y \mapsto B^T y$ on this domain. As a result, the amplitudes depend only on the weight $|y|$,
and the resulting transformation corresponds to applying a
degree-$\ell$ polynomial $P$ to the objective value $f(x)$. 
    \item Apply $H^{\otimes n}$ to get $\ket{\psi_P}$ in Eq.~\ref{eqn:dqi_state}.
\end{enumerate}

The output of the algorithm is obtained by measuring
$\ket{\psi_P}$ in the computational basis. The choice of polynomial $P$ controls how amplitudes depend on the
objective value $f(x)$. Note that $P$ is not specified explicitly, but implicitly defined
by the choice of weights $w_k$. In particular, the resulting amplitudes
depend only on $f(x)$ and can be written as a degree-$\ell$ polynomial
function $P(f(x))$. We can write down the DQI state used in Max-XORSAT explicitly as a Fourier transform: 
\begin{align}
\label{eq:explicit_max_xorsat}
    \ket{\psi_P} &= H^{\otimes n} \sum_{k=0}^\ell \frac{w_k}{\sqrt{m \choose k}} \sum_{y \in \mathbb{F}_2^m, |y|=k} (-1)^{v \cdot y} \ket{B^T y},
\end{align}
where $\sum_{k=0}^\ell |w_k|^2 = 1$.
DQI extends naturally to linear constraint satisfaction over finite fields.

\begin{problem}[Max-LINSAT Search]
\label{problem:linsatsearch}
Given $B \in \mathbb{F}_p^{m\times n}$ and subsets
$F_1,\dots,F_m \subseteq \mathbb{F}_p$,
find $x \in \mathbb{F}_p^n$ maximizing the number of satisfied
constraints $(Bx)_i \in F_i,\, (1\le i\le m)$.
Equivalently, maximize:
\begin{align}
\label{eqn:def_linsat_cost}
f(x)
=
\sum_{i=1}^m f_i(b_i \cdot x),
\end{align}
where
\[
f_i(t) =
\begin{cases}
1 & t \in F_i,\\
-1 & t \notin F_i .
\end{cases}
\]
\end{problem}

Similarly to Max-XORSAT, the algorithm starts by preparing a DQI state that amplifies the output probability. We have that 
\begin{align}
    \ket{\psi_P} &= H_p^{\otimes n}  \sum_{k=0}^\ell w_k \ket{P^{(k)}} = H_p^{\otimes n} \sum_{k=0}^\ell \frac{w_k}{\sqrt{m \choose k}} \sum_{{\bf y}\in \mathbb{F}^m_p; |{\bf y}| = k} \prod_{i=1; y_i \neq 0}^m \tilde{g}_i(y_i) \ket{B^T {\bf y}}\,,
    \label{eq:maxlinsat_state}
\end{align}
where $H_p$ is the Fourier transform over $\mathbb{F}_p$, $\tilde{g}_i$ is a Fourier transform of a shifted and rescaled $f_i$ (see \cite[Section 8.2]{dqi}), and $\sum_{k=0}^\ell |w_k|^2 = 1$.

\subsubsection{Optimal Polynomial Intersection} 
OPI is a special case of Max-LINSAT where $p=m-1$ and $B$ generates a Reed-Solomon code. In \cite{dqi}, the authors show that DQI with the Berlekamp-Massey decoder achieves a higher satisfying fraction of constraints compared to Prange's algorithm, the best classical algorithm known to the authors. This points to the possibility of DQI achieving a superpolynomial speedup for this problem.
%%\pagebreak

\section{Searching for peaks in DQI distributions}
\label{sec:peakfinding}
\subsection{Overview}
We show that all outcomes of a DQI distribution with probability at least
$1/\poly(n)$ (``peaks'') can be efficiently found by a classical algorithm. The argument proceeds in three steps. First, we show that the DQI state is \emph{computationally
tractable} (CT) in the sense of \cite{vdn09}. Second, we use the results of
Schwarz and van den Nest~\cite{svdn13} to show that a sequence of telescoping marginals of the DQI distribution can be estimated to
inverse-polynomial additive error. We then use a peak-finding
algorithm from learning theory
\cite{goldreich1989hard,kushilevitz1991learning,svdn13}
to find all inverse-polynomially large probabilities in the distribution. 

While it may seem like this would allow us to classically approximately sample from the DQI output distribution, we argue that this is not the case. We show that for typical Max-LINSAT instances, the output distribution is flat and no outcome carries inverse-polynomial probability mass. So, this simulation strategy for DQI fails. It is interesting that the same situation occurs in Shor's algorithm, as highlighted in Ref.~\cite{svdn13}.

\subsection{Computational Tractability}
\begin{definition}[{Computational Tractability~\cite{vdn09}}]
An $n$-qubit quantum state $\ket{\psi}$ is \emph{computationally tractable} (CT) if: 
\begin{enumerate}
\item  For every $x\in\{0,1\}^n$, one can compute $\poly(n)$ many bits of $\braket{x}{\psi}$ in $\poly(n)$ time.
\item the distribution $\{x \textnormal{ w.p. } |\braket{x}{\psi}|^2\}$ can be classically sampled from in  $\poly(n)$ time.
\end{enumerate}
\label{Def:CT}
\end{definition}
The main technical tool is the following result, which shows that the overlap between two CT states can be efficiently estimated with additive error. For convenience, we provide a proof here:

\begin{lemma}[Overlap Lemma {\cite[Lemma 19]{svdn13}}]
Let $\ket{\psi}$ and $\ket{\phi}$ be CT states such that amplitudes
$\braket{x}{\psi},\braket{x}{\phi}$ are computable in time $C_\psi,C_\phi$
and samples from $|\braket{x}{\psi}|^2,|\braket{x}{\phi}|^2$ 
can be generated in time $S_\psi,S_\phi$.
Then $\braket{\psi}{\phi}$ can be estimated to additive error $\epsilon$
with probability $1-\delta$ in time:
\[
O\!\left(\frac{(C_\psi+C_\phi)(S_\psi+S_\phi) \cdot \log 1/\delta}{\epsilon^2}\right).
\]
\label{lemma:overlap}
\end{lemma}
\begin{proof}
Rewrite the overlap as:
\[
\braket{\psi}{\phi}
=
\mathbb{E}_{x\sim |\braket{x}{\psi}|^2}
\!\left[
\frac{\braket{x}{\phi}}{\braket{x}{\psi}}
I_{|\braket{x}{\psi}|^2 \ge |\braket{x}{\phi}|^2}
\right]
+
\mathbb{E}_{x\sim |\braket{x}{\phi}|^2}
\!\left[
\frac{\braket{\psi}{x}}{\braket{\phi}{x}}
I_{|\braket{x}{\psi}|^2 < |\braket{x}{\phi}|^2}
\right],
\]
where $I_E$ is an indicator on event $E$. 
Both estimators are bounded in $[-1,1]$ and can be evaluated using
the ability to compute amplitudes and sample from the corresponding
Born distributions $|\braket{x}{\psi}|^2, |\braket{x}{\phi}|^2$. By the Chernoff–Hoeffding bound,
$O(1/\epsilon^2 \cdot \log 1/\delta)$ samples suffice to estimate each expectation to
additive error $\epsilon$ with success probability $1 - \delta$.
The stated running time follows from the cost of sampling and
evaluating the amplitudes.
\end{proof}
 The following lemma shows that marginal probabilities of CT states can be estimated efficiently. The key idea is to express the marginal as the overlap of two CT states and apply \Cref{lemma:overlap}.
\begin{lemma}[\cite{svdn13}]
\label{lemma:ct_marginals}
    Let $\ket{\psi}$ be a CT state on $n$ qudits. Let $k \in \lbrace 0, \ldots n \rbrace$, $y \in \lbrace 0, \ldots d-1 \rbrace^{n-k}$, and let $M = \sum_{x \in \lbrace 0, \ldots d-1 \rbrace^k} \ket{x}\bra{x} \otimes \ket{y} \bra{y} $. Define: 
    \begin{align}
        H_d &= \frac{1}{\sqrt{d}} \sum_{x, k \in \lbrace 0, \ldots d-1 \rbrace} \exp\left[-\frac{2\pi i kx}{d}\right] \ket{x}\bra{k}.
    \end{align}Then there is a randomized classical algorithm that w.h.p. approximates $\bra{\psi} H_d^{\dagger \otimes n} M H_d^{\otimes n}  \ket{\psi}$
    to $\epsilon$ additive error in time that is polynomial both in $n$ and in $1/\epsilon^2$.
\end{lemma}
\begin{proof}
    We write the marginal projector as a swap test
between two appropriately transformed states (see \Cref{fig:tensoridentity}) to represent $\bra{\psi} H_d^{\dagger \otimes n} M H_d^{\otimes n}  \ket{\psi}$ as the overlap of states $H_d^{n-k}\ket{\psi,y}$ and $A H_d^{n-k} \ket{\psi,y}$, where $A$ swaps the last $n-k$ qudits with the previous $n-k$ qudits.
\begin{figure}[ht]
  \centering
  \includesvg[width=\textwidth]{hadamard_identity_rhs_v11}
    \caption{\small Expressing a marginal as an overlap of two CT states. 
    } 
    \label{fig:tensoridentity}
\end{figure}

We show that both $H_d^{\otimes n-k} \ket{\psi, y}$ and $A H_d^{\otimes n-k} \ket{\psi, y}$ states are CT:
    \begin{itemize}
        \item For $ H_d^{\otimes (n-k)}\ket{\psi,y}$, any amplitude: 
        \begin{align}
        \bra{x,z}H_d^{\otimes (n-k)}\ket{\psi,y} &= \bra{x}\ket{\psi} \frac{1}{\sqrt{d^{n-k}}}  \exp\left[ \frac{-2i\pi z \cdot y}{d}\right], \end{align}
        can be efficiently queried. The  distribution can be written as $\{(x,z) \textnormal{ w.p. } |\braket{x,z}{\psi,y}|^2\} = \{(x,z) \textnormal{ w.p. } |\braket{x}{\psi}|^2/d^{n-k}\}$. This can be sampled from by first drawing $x$ from $|\braket{x}{\psi}|^2$  (since $\ket{\psi}$ is a CT state) and then drawing $z \in \lbrace 0, \ldots d-1 \rbrace^{n-k}$ uniformly at random.
        \item Now we inspect $A H^{\otimes (n-k)}\ket{\psi,y}$. The amplitude at a string $x \defeq x_1 || x_2$ can be written as\footnote{$||$ denotes concatenation}:
        \begin{align} \bra{x_1,x_2,z}A\ket{\psi,y} = \bra{x_1,z}\ket{\psi} \exp\left[\frac{2i\pi x_2 \cdot y}{d}\right] \frac{1}{\sqrt{d^{n-k}}}. \end{align} The distribution is then $\{(x_1,x_2,z) \textnormal{ w.p. } |\braket{x_1,z}{\psi}|^2 / d^{n-k}\}$.  
        This can be sampled from by drawing $x_1,z$ from $|\braket{x_1,z}{\psi}|^2$ and then $x_2 \in \lbrace 0, \ldots d-1 \rbrace^{n-k}$ uniformly at random. 
    \end{itemize}
    The proof then follows by  \Cref{lemma:overlap}. 
\end{proof}
\begin{remark}
\label{cor:ct_hadamard_marginals}
Let $\ket{\psi}$ be a CT state. Then the marginal distribution over the
first $m$ qudits of the Born distribution
$\{x \text{ w.p. } p(x) = |\bra{\psi} H_d^{\otimes n} \ket{x}|^2\}$
can be approximated to inverse-polynomial additive error in polynomial
time.
\end{remark}

Computing the marginals of $\ket{\psi}$ allows one to find all $1/\poly(n)$-peaked probabilities in $|\braket{x}{\psi}|^2$. 
\begin{theorem}[{Peak-finding algorithm \cite{goldreich1989hard,kushilevitz1991learning,svdn13}}]
\label{thm:peakfinding}
Let $p$ be a probability distribution on $\{0,1, \ldots d-1\}^n$, and let:
\begin{align}
p_m(y) &= \sum_{x\in\{0,\ldots,d-1\}^m} p(x||y) \end{align} be the distribution of $p$ after marginalizing over the first $m$ bits. Suppose there is an efficient classical algorithm that for any $0 \le m \le n$ and $y \in \lbrace 0,\ldots,d-1 \rbrace^{n-m}$, w.h.p. computes $p_m(y)$ to $1/\poly(n)$ additive error. Then, for any threshold $\theta > 0$ and failure probability $\pi > 0$, there is a randomized classical $\poly(n, 1/\theta, \log(1/\pi))$-time algorithm that outputs a set of bitstrings $L$ with $|L| \le 2/\theta$, and with probability $1 - \pi$, the below conditions are satisfied:
\begin{enumerate}
    \item For all $y \in L$, we have $p(y) \geq \theta/2$.
    \item For all bitstrings $x \in \{0,1\}^n$, if $p(x) \ge \theta$ then $x \in L$.
\end{enumerate}
\end{theorem}
By \Cref{thm:peakfinding}, for any CT state $\ket{\psi}$, we can efficiently find peaks from the distribution induced by $H_d^{\otimes n} \ket{\psi}$, i.e. $\{x \text{ w.p. } p(x) = |\bra{\psi} H_d^{\otimes n} \ket{x}|^2\}$.

\subsection{Peak-finding algorithm}
\label{subsec:searching}
We now show that the Hadamard transform of the DQI state is CT, which allows us to apply the above machinery.
\begin{claim}
\label{claim:hadamard_dqi_is_ct}
Consider a Max-XORSAT DQI state $\ket{\psi_P}$ from \cref{eqn:dqi_state}, prepared with a \emph{classical} decoder. Then the state $H_2^{\otimes n} \ket{\psi_P}$ is CT.
\end{claim}
\begin{proof}
Let $\ket{\psi'_P} \defeq H_2^{\otimes n} \ket{\psi_P}$. Using Eq.~\ref{eq:explicit_max_xorsat}, this state  has the form
    \begin{align*}
        \ket{\psi'_P} = \sum_{k=0}^\ell  \frac{w_k}{\sqrt{ { m \choose k}}} \sum_{y \in \mathbb{F}_2^m, |y| = k} (-1)^{v \cdot y} \ket{B^T y}\,,
    \end{align*}
    for a known choice of weights $(w_0, w_1, \dots, w_\ell)$ such that $\sum_{k=0}^\ell |w_i|^2 = 1$. Any $y \in \mathbb{F}_2^m$ where $|y| \le \ell$ can be uniquely decoded from $B^T y$, so each basis state $\ket{B^T y}$ in $\ket{\psi'_P}$ is distinct. We verify CT-ness:
    \begin{enumerate}
        \item Given a basis state $\ket{s}$ where $s = B^T y$, we can use the DQI decoder to \emph{classically} infer $y$. The state $\ket{\psi'_P}$ is supported on $\ket{s}$ if and only if the decoder successfully infers some $y$ where $|y| \le \ell$. If the decoder fails, $\bra{s}\ket{\psi'_P} = 0$. Otherwise, $\bra{s}\ket{\psi'_P} = \frac{w_{|y|}}{\sqrt{ { m \choose |y| }}} (-1)^{v \cdot y}$.
\item To sample from the distribution  $\{x \textnormal{ w.p. } |\braket{x}{\psi'_P}|^2\}$, we first sample $k \in [0,1, \dots, \ell]$ according to  probability distribution $(|w_0|^2, |w_1|^2, \dots, |w_\ell|^2)$. Then, sample a uniformly random subset $S \subseteq [m]$ of size $|S| = k$. Let $y \in \mathbb{F}_2^m$ be the indicator vector of $S$; then return $B^T y$.\qedhere
    \end{enumerate}
\end{proof}

\begin{corollary}
\label{cor:peakfind_dqi}
Consider a DQI state $\ket{\psi_P}$ from \cref{eqn:dqi_state} prepared with a \emph{classical} decoder, and fix a polynomial $q$. Then there is a randomized classical $\poly(n)$-time algorithm that can w.h.p. output a list containing all bitstrings $s$ such that $|\bra{s}\ket{\psi_P}|^2 \ge 1/{q(n)}$.
\end{corollary}
We remark that the existence of an efficient peak-finding algorithm does not
imply that sampling from $\ket{\psi_P}$ is classically tractable.
This phenomenon is analogous to Simon's algorithm, Shor's algorithm,
and IQP circuits~\cite{vdn09,svdn13}, where heavy outputs can be
identified efficiently despite the conjectured hardness of sampling. 
We investigate the complexity of sampling from DQI distributions in
Sec.~\ref{sec:stockmeyer}.

\subsection{Are there peaks in DQI distributions?}

We proved in \Cref{subsec:searching} that whenever a DQI distribution contains a $1/\poly(n)$ peak, it can be located in polynomial time classically. We now study when DQI distributions contain peaks. First, we give examples of peaked DQI states. We then isolate a large family of DQI states for which the output distributions do not peak and thus resist the above classical simulation strategy. 

\subsubsection{Peaked DQI states}
\label{subsec:peaked}
We give two constructions of peaked DQI states.

\begin{lemma}
For any bitstring $x \in \{0,1\}^n$, there is a DQI state $\ket{\psi_P} = \ket{x}$
when $C^\perp = \lbrace 0^m, 1^m \rbrace$ is the repetition code.
\end{lemma}
\begin{proof}
In $C^\perp$, the distance $d^\perp = m$ is maximal. We choose $m = 2t+1$ odd, so $t < d^\perp/2$. Then:
    \begin{align}
        N_t := \sum_{k=0}^t {m \choose k} &= 2^{m-1}\,.
    \end{align}
    Since $|C| = 2^{m-n} = 2$, we have $m = n+1$, and so $N_t = 2^{m-1} = 2^n$. 
    The DQI state in Eq.~\ref{eqn:dqi_state} is:
     \begin{align}
         \ket{\psi_P} = H^{\otimes n} \sum_{k=0}^{t} \frac{w_k}{\sqrt{{m \choose k}}} \sum_{|y| = k}(-1)^{v \cdot y} \ket{B^Ty}
     \end{align}
    One such state uses weights $w_k = \sqrt{{m \choose k}/2^n}$:
    \begin{align}
        \ket{\psi_P} &= H^{\otimes n} \sum_{|y| \leq t} \frac{(-1)^{v \cdot y}}{\sqrt{2^n}}\ket{B^T y}.
    \end{align}
    Note that
    \begin{align}
        \braket{x}{\psi_P} &= 2^{-n} \sum_{|y| \leq t} (-1)^{y \cdot(v+Bx)} 
    \end{align}
    When $v = Bx$, then this value is $1$. Here, $\ket{\psi_P} = \ket{x}$,
   i.e. concentrating on a single outcome.
\end{proof}
\begin{lemma}
    For any amplitude $x \in \{0,1\}^n$, there is a DQI state $\ket{\psi_P} = \ket{x}$ when $C^\perp$ is the $[2^n-1, 2^{n}-n-1, 3]_2$ Hamming code.
\end{lemma}
\begin{proof}
    Since the distance of $C^\perp$ is $3$, we can correct at most one error, and
    $$
    \sum_{k=0}^1 {m \choose k} = m + 1 = 2^n\,.
    $$
    Then the DQI state in Eq.~\ref{eqn:dqi_state} is:
    $$
    \ket{\psi_P} = H^{\otimes n} \sum_{k=0}^1 \frac{w_k}{\sqrt{m \choose k}} \sum_{|y|=k} (-1)^{v \cdot y} \ket{B^T y}\,.
    $$
    One such state uses weights $w_k = \sqrt{{m \choose k}/2^n}$:
    $$
    \ket{\psi_P} = H^{\otimes n} \sum_{|y| \le 1} \frac{(-1)^{v \cdot y}}{\sqrt{2^n}} \ket{B^T y}\,.
    $$
    Note that
    $$
    \braket{x}{\psi_P} = 2^{-n} \sum_{|y| \le 1} (-1)^{y \cdot (v + Bx)}\,.
    $$
    When $v = Bx$, then this value is $1$. Here, $\ket{\psi_P} = \ket{x}$,
   i.e. concentrating on a single outcome.
\end{proof}

In both of these examples, peaks can be located by the algorithm in \Cref{subsec:searching}.
In fact, DQI does not offer quantum advantage on these instances.
In the first example, Prange's algorithm already finds the optimal solution.
In the second example, $m$ is superlinear in $n$, and so DQI finds solutions satisfying only $o(m)$ constraints more than random guessing. 
\begin{remark}
    When $m \le n + 1$, the problem can be solved optimally with Prange's algorithm. Choose any $n$ linearly independent constraints; there is exactly one vector $x \in \{0,1\}^n$ satisfying these constraints, which can be found by diagonalizing the constraint matrix. Then $x$ is optimal: if a solution satisfies all constraints, then $x$ will satisfy the remaining constraint.
\end{remark}
\begin{remark}
    Suppose $m$ is superlinear in $n$. We upper-bound the fraction of constraints that DQI is expected to satisfy. Recall that $\ell \le n$; then by \cite[Eq. 9]{dqi}, this is
    $$
    \left( \sqrt{\frac{\ell}{m} (1 - \frac{r}{p})} 
    + \sqrt{(1-\frac{\ell}{m}) \frac{r}{p}}
    \right)^2
    =  
    \left(o(1) + \sqrt{(1 - \frac{\ell}{m})\frac{r}{p}} \right)^2
    \le \frac{r}{p} + o(1)\,.
    $$
    So DQI is expected to satisfy the same asymptotic fraction of constraints as in random guessing.
\end{remark}

\subsubsection{DQI states with no peaks}
Here we show that DQI states for more interesting problem instances than those in Subsec.~\ref{subsec:peaked} produce output distributions with no $1/\poly(n)$ peaks. 

\begin{lemma}[Flatness criterion for DQI states]
\label{thm:flatness_criterion}
Let $\ket{\psi}$ be a Max-LINSAT DQI state (Eq.~\ref{eq:maxlinsat_state}) over $\mathbb{F}_p$, obtained by Fourier-transforming a state of the form
\[
\ket{\phi}
=
\sum_{\substack{z\in\mathbb{F}_p^m\\ |z|\le \ell}} c_z \ket{B^T z},
\qquad
\sum_{|z|\le \ell} |c_z|^2 = 1,
\]
where the map $z\mapsto B^T z$ is injective on $\{z:|z|\le \ell\}$.
Let
\[
N_\ell
:=
\#\{z\in\mathbb{F}_p^m:|z|\le \ell\}
=
\sum_{k=0}^{\ell}\binom{m}{k}(p-1)^k.
\]
Then for every $x\in\mathbb{F}_p^n$,
\[
|\langle x|\psi\rangle|^2 \le p^{-n} N_\ell.
\]
\end{lemma}

    \begin{proof}
Since the decoder corrects up to $\ell$ errors, the map
$z \mapsto B^T z$ is injective on $\{z : |z|\le \ell\}$.
Applying the Fourier transform $ H_p^{\otimes n}$ to $\ket{\phi}$ yields:
\[
\braket{x}{\psi}
=
p^{-n/2}
\sum_{|z|\le \ell} c_z \,\omega^{z\cdot (Bx)},
\qquad
\omega = e^{2\pi i/p}.
\]
Therefore
\[
|\braket{x}{\psi}|
\le
p^{-n/2}
\sum_{|z|\le \ell} |c_z|.
\]
By Cauchy--Schwarz,
\[
\left(\sum_{|z|\le \ell} |c_z|\right)^2
\le
N_\ell \sum_{|z|\le \ell} |c_z|^2
=
N_\ell.
\]
Hence
\begin{align*}
|\braket{x}{\psi}|^2 \le p^{-n} N_\ell. \tag*{\qedhere}
\end{align*}
\end{proof}

In particular, if there exists $\delta>0$ such that $N_\ell \le p^{(1-\delta)n}$,
then
the output distribution has no inverse-polynomially large peaks. We rely on the standard fact about Hamming ball volume and two fundamental bounds that relate the rate and distance of a code.
\begin{fact}[Hamming ball volume entropy bound \cite{van1992introduction}]
\label{fact:entropy_boud}
Let $N_\ell$ be the volume of a Hamming ball of size $\ell$ in $\mathbb{F}^m_p$. Then
    \begin{align}
        \log_p N_\ell \leq m H_p \left( \frac{\ell}{m} \right) + o(m).
    \end{align}
\end{fact}
\begin{fact}[Elias-Bassalygo]
\label{fact:elias}
At large enough $n$, every $[n,k,d]_q$ code $C$ satisfies the following inequality for all $0 \le d/n \le (1 - 1/q)$:
$$
\frac{k}{n} \le 1 - H_q \left(J_q\left(\frac{d}{n}\right)\right) + o(1)\,,
$$
where 
$$
J_q(x) = (1-1/q) \cdot \left( 1-\sqrt{1-\frac{qx}{q-1}}\right)\,,
$$
and
$H_q$ is the $q$-ary entropy function:
\[
H_q(x)
=
x\log_q(q-1)
-x \log_q x
-(1-x)\log_q(1-x)\,.
\]
\end{fact}

\begin{fact}[Plotkin]
\label{fact:plotkin}
    Every $[n,k,d]_q$ code $C$  with $d/n > 1 - 1/q$ satisfies the following inequality:
    $$
    q^k  \le \frac{qd}{qd - (q-1)n}\,.
    $$
\end{fact}

\begin{theorem}
\label{prop:entropy_gap_regimes}
Let $C^\perp$ be an $[m,m-n,d^\perp]_p$ code and suppose $\ell < d^\perp/2$. Then:

\begin{enumerate}
\item When $0 < d^\perp/m \le 1 - 1/p$: $N_{\ell} \le p^{n - m(\Delta - o(1))}$. Here $\Delta > 0$ for all $p\ge 2$. 

\item When $d^\perp/m \to 0$, then $N_\ell = p^{o(m)}$. 

\item When $d^\perp/m > 1 - 1/p$: $N_{\ell} \le p^{n - m(\Delta - o(1))}$.  Here $\Delta > 0$ when $p>2$, or when $d^\perp/m \nrightarrow 1$.
\end{enumerate}
\end{theorem}

\begin{proof}
We consider regimes of the relative distance $\delta^\perp := d^\perp/m$.

\begin{enumerate}
\item For $\delta^\perp \le 1 - 1/p$, \Cref{fact:elias} implies
$$
\frac{m-n}{m} \le 1 - H_p(J_p(\delta^\perp)) + o(1) \implies \frac{n}{m} \ge H_p(J_p(\delta^\perp)) - o(1)\,.
$$
Let $\Delta := H_p(J_p(\delta^\perp)) - H_p(\delta^\perp/2)$. By inspection, $\Delta >0$ for $p \ge 2$ and positive value of $\delta^\perp$ in this region. Then
$$
\frac{n}{m} \ge H_p(\delta^\perp/2) + \Delta - o(1)\,.
$$
From Fact~\ref{fact:entropy_boud}, we have $\log_p N_{\ell} \le m H_p(\ell/m) + o(m)$. Since $H_p$ is increasing in this region and $\ell < d^\perp/2$, we have
$$
\log_p N_{\ell} \le n - m \Delta + o(m)\,.
$$
This implies 
$$
N_{\ell} \le p^{n - m(\Delta - o(1))}\,.
$$

\item Suppose $\delta^\perp \to 0$.
Then $\ell/m \to 0$. Since $H_p(x) = O(x \log(1/x))$ as $x \to 0$, it follows from Fact~\ref{fact:entropy_boud} that
$\log_p N_\ell = o(m)$.
Thus $N_\ell = p^{o(m)}$. This is exponentially smaller than $p^n$ provided $m = O(n)$. 

\item Now we assume $d^\perp/m > 1 - 1/p$. \Cref{fact:plotkin} implies
$$
p^{m-n} \le \frac{p d^\perp}{p d^\perp - (p-1)m}\,.
$$
The right-hand side is at most $p d^\perp \le p m$, because the denominator is a positive integer. This implies
\begin{align}
\label{eq:detail_pt3}
m - n \le 1 + \log_p m \le o(m)\,.
\end{align}
We again consider the logarithm of $N_{\ell}$. Using the same steps as before,
$$
\log_p N_{\ell} \le m H_p(d^\perp/(2m)) + o(m)\,.
$$
Let $\Delta = 1 - H_p(d^\perp/(2m))$. Then using \cref{eq:detail_pt3},
$$
\log_p N_{\ell} \le m - m \Delta + o(m) \le n - m \Delta + o(m)\,.
$$
This implies
$$
N_{\ell} \le p^{n - m(\Delta - o(1))}\,.
$$
By inspection, $\Delta > 0$ for any value of $d^\perp/m$ in this range, \textit{except} at $p = 2$ and $d^\perp/m \to 1$.
In this case, \Cref{fact:elias} implies $2^{m-n} \le 2 + o(1)$, and so $m \le n + 1$. \qedhere
\end{enumerate}
\end{proof}
\Cref{prop:entropy_gap_regimes}, when combined with Lemma~\ref{thm:flatness_criterion}, isolates the regime in which DQI distributions contain no $1/\poly(n)$ peaks.
\begin{corollary}
\label{cor:flatness}
Let $\ket{\psi}$ be a Max-LINSAT DQI state (Eq.~\ref{eq:maxlinsat_state}) over $\mathbb{F}_p$.
Then for every $x\in\mathbb{F}_p^n$, $|\braket{x}{\psi}|^2 \le p^{-cn}$ for some constant $c > 0$, except if $m \le n + 1$, or $m$ is superlinear in $n$.
\end{corollary}

\begin{remark}[Sharpness and perfect codes] 
Suppose the dual code $C^\perp$ is a perfect code, with decoding radius
\[
t=\left\lfloor \frac{d^\perp-1}{2}\right\rfloor,
\]
and take $\ell=t$. Then the Hamming balls of radius $t$ around codewords of $C^\perp$
partition $\mathbb{F}_p^m$, so
\[
|C^\perp|\,N_t = p^m.
\]
Since $C^\perp$ is an $[m,m-n,d^\perp]_p$ code, we have
\[
|C^\perp| = p^{m-n},
\]
and therefore
\[
N_t = p^n.
\]
In this case, Lemma~\ref{thm:flatness_criterion} yields only the trivial bound
\[
|\langle x|\psi\rangle|^2 \le 1.
\]
Here, DQI states can exhibit peaked output distributions, as in the constructions
of Subsec.~\ref{subsec:peaked}.
\end{remark}

Results of Subsection~\ref{subsec:searching} show that if a distribution
admits outcomes with inverse-polynomial probability, then such
``peaked'' outcomes can be efficiently identified using the
Kushilevitz--Mansour algorithm and its extensions~\cite{svdn13}.
At first sight, this suggests a possible classical simulation strategy
for DQI. However, the above analysis shows that DQI Max-LINSAT output distributions are, bar a few exceptions, inherently flat: every outcome has probability at most
$p^{-\Delta n}$ for some constant $\Delta>0$. In particular, there are no
inverse-polynomially heavy outcomes and the peak-finding framework is inapplicable in this setting.
Rather than contradicting the possibility of classical simulation,
our results identify a structural obstruction: any successful
dequantization of DQI must exploit features other than the presence
of heavy Fourier components or large output probabilities. It is interesting that this almost exactly mirrors the situation of Shor's algorithm studied in \cite{svdn13}. 

The typical ``flatness'' of the output distributions suggests that they may be \emph{anticoncentrated}. This raises the question of whether one can use the usual Aaronson-Arkhipov construction \cite{aaronson2011computational} to argue for hardness of sampling from DQI distributions. In \Cref{sec:stockmeyer} we outline obstructions to such a hardness of sampling argument.

\subsubsection{Outputs of DQI are not too rare}
We remark that since typical DQI distributions lack peaks, the bitstrings output by DQI are not rare. For essentially any output of the DQI distribution, there are \textit{exponentially many} codewords with solution quality matching this output. This implies,  for example, that DQI cannot solve CSP instances with a small number of planted satisfying assignments.
\begin{claim}[No ``unusually good'' outcomes from DQI]
\label{claim:unusuallygooddqi}
Suppose $m > n+1$ and $m$ is linear in $n$.
For any DQI distribution $\mathcal{D}$ over $\mathbb{F}_p$, there exists a $c > 0$, where an output $x$ sampled from $\mathcal{D}$ has $p^{cn}$ other outputs $x'$ with the same score as $x$, except with at most $p^{-cn}$ probability.
\end{claim}
\begin{proof}
Let the score of any output $x$ be $f(x)$; this is an integer between $0$ and $m$. Let $S_z$ be the set of outcomes $x$ with score $z = f(x)$. A DQI state has the same amplitude for any outcome with the same score; let $q(z)$ be the probability when the score is $z$. Then
$$
1 = \sum_{z=0}^m |S_z| q(z)\,.
$$
Let $\mathcal{N}_t$ be the set of $z$ where $|S_z| p(z) \le \frac{1}{t(n)}$, for some function $t$ we choose later. Then
$$
\sum_{z \notin \mathcal{N}_t} |S_z| q(z) = 1 - \sum_{z \in \mathcal{N}_t}|S_z| q(z) \ge 1 - \frac{m}{t(n)}\,. 
$$
With $\ge 1 -\frac{m}{t(n)}$ probability, we sample an outcome $x$ with score $f(x) = z \notin \mathcal{N}_t$. Then
$$
|S_z| q(z) \ge t(n)\,,
$$
but by \Cref{cor:flatness}, there is a $\delta > 0$ where $q(z) \le p^{-\delta n}$. This means the size of $|S_z|$ is large; i.e.
\begin{align*}
|S_z| \ge p^{\delta n}/t(n)\,.
\end{align*}
We choose $t(n) = p^{n \cdot 2 \delta/3}$; then since $m$ is polynomial in $n$, the claim holds for $c \defeq \delta / 3$.
\end{proof}
\Cref{claim:unusuallygooddqi} means that DQI samples from an exponentially large set of interesting solutions.
Since this set is so large, the distribution is ``flat'', and therefore hard to simulate by the techniques in this section.
This is similar to the situation in Shor's algorithm,
where one hides an exponentially large set $S$ in a quantum state, so that sampling the state allows the recovery of information about $S$. In Shor's algorithm, $S$ is a congruence class determined by the discrete log value (see \Cref{sec:simulate_shor}). In DQI, $S$ is essentially the intersection of codewords $\{Bx \ | \ x \in \mathbb{F}_2^n\}$ and vectors at some distance from $v$, but without the associated group structure. 

% \pagebreak

\section{Sampling from DQI distributions}
\label{sec:stockmeyer}

\subsection{Overview}
A central theme in sampling-based quantum advantage is the relationship
between sampling from a distribution and approximating its output
probabilities. A standard ``quantum supremacy'' argument~\cite{aaronson2011computational} proceeds in two steps:
\begin{enumerate}
    \item  Assuming an efficient classical sampler, use Stockmeyer's theorem
to multiplicatively approximate output probabilities in $\BPP^\NP$.
\item Show (or conjecture) that such multiplicative approximation is $\#\P$-hard. Combined with (1), this implies collapse of the polynomial hierarchy, which is unlikely. This is treated as evidence that efficient classical sampling is hard. 
\end{enumerate}

This idea underlies hardness arguments for BosonSampling,
IQP circuits~\cite{bremner2011classical}, and random circuit
sampling~\cite{bouland2019complexity}. We first remark that this proof strategy does not work for DQI because the
DQI output probabilities are computable in $\FP$. This precludes the contradiction used in the standard hardness-of-sampling argument. We strengthen this by showing that the full DQI output
distribution can be approximately sampled in $\BPP^\NP$. Together,
these results rule out the standard Stockmeyer-based route to
establishing classical hardness of DQI.

\subsection{DQI output probabilities in \texorpdfstring{$\FP$}{FP}}
We first show that individual output probabilities of DQI states can
be computed classically in polynomial time. This simple observation implies obstruction to the usual hardness-of-sampling argument to DQI.
\begin{lemma}
    Let $\ket{\psi_P}$ be a DQI state (Eq.~\ref{eqn:dqi_state}) and assume, as in the DQI state preparation procedure from \cite{dqi}, that the weights $\lbrace w_k \rbrace_{k=0}^\ell$ are known.
There exists a polynomial-time classical algorithm that computes $\braket{x}{\psi_P}$.
\label{lemma:g_classical}
\end{lemma}
\begin{proof} 
We first prove the lemma for Max-XORSAT state (Eq.~\ref{eqn:dqi_state}).
Following~\Cref{sub:setup}, the Hadamard transform of $\ket{\psi_P}$ is given by: 
\begin{align}
    H^{\otimes n} \ket{\psi_P}  &= \sum_{k=0}^{\ell} \frac{w_k}{\sqrt{{m \choose k}}} \sum_{y; |y|=k} (-1)^{v\cdot y} \ket{B^T y},
\end{align}
where we assume that $w_k$ are known and easily computable. Let $a(x) = Bx + v$ and  evaluate the contribution of each weight-$k$ slice separately: 
\begin{align}
    \braket{x}{\psi_P} &= \bra{x} H^{\dagger \otimes n} H^{\otimes n} \ket{\psi_P} \\&=  \sum_{k=0}^\ell  \frac{w_k}{\sqrt{{m \choose k}}} \sum_{y; |y| = k} (-1)^{v\cdot y} \braket{x}{H^{\otimes n} | B^T y} \\ &=  \sum_{k=0}^\ell  \frac{w_k}{\sqrt{{m \choose k}}} \frac{1}{\sqrt{2^n}} \sum_{y; |y| = k} (-1)^{(v + Bx) \cdot y} \\ &=  \sum_{k=0}^\ell  \frac{w_k}{\sqrt{{m \choose k}}}\frac{1}{\sqrt{2^n}} \sum_{y; |y| = k} (-1)^{a(x) \cdot y}
\end{align}
Note that: 
\begin{align}
    \sum_{y; |y| = k} (-1)^{a(x) \cdot y} = \sum_{j=0}^k (-1)^j {|a(x)| \choose j}{m - |a(x)| \choose k - j}
\end{align}
This is computable in polynomial time; in fact, it is an evaluation of the Kravchuk polynomial. Since we know $w_k$, we need to compute up to $\ell$ such polynomials.
\medskip 

For Max-LINSAT, using \cite[Section 8.2]{dqi} and its assumptions, the amplitude of the Max-LINSAT state in Eq.~\ref{eqn:dqi_state}, $\bra{x}\ket{\psi_P}$, can be written as:
$$
\sum_{k=0}^\ell \frac{w_k}{\sqrt{p^{n-k} {m \choose k}}} P^{(k)}(g_1( (Bx)_1 ), \dots, g_m( (Bx)_m))\,,
$$
where $g_i$ is efficiently computable, and $P^{(k)}$ is the $k^\text{th}$ elementary symmetric polynomial. This can be evaluated in $\poly(m,p)$ time with Newton's identities.
\end{proof}

The above lemma rules out the standard Aaronson--Arkhipov route to sampling hardness for DQI \cite{aaronson2011computational}: multiplicative approximation of output
probabilities cannot
be $\#\P$-hard here. Any argument that would support hardness of classical simulation of DQI must therefore
exploit an alternative proof strategy.

One may wonder if the \textit{marginal} output probabilities are in $\P$; i.e. if $\sum_{x \in S} p(x)$ can be efficiently computed for any $S \subseteq \{0,1\}^n$. If this was the case, then DQI could be classically simulated, as we could bit-by-bit sample from the output distribution. However, we do not know how to compute the marginal output probabilities in $\P$ or even $\BPP$. Nonetheless, these probabilities can be computed to multiplicative error in $\BPP^{\NP}$ using Stockmeyer's theorem. In the next subsection, we introduce this $\BPP^{\NP}$ algorithm, and explain how it can sample from the output distribution of DQI.

\subsection{Approximate sampling in \texorpdfstring{$\BPP^\NP$}{BPP with NP oracle}}

We use the following formulation of Stockmeyer's approximate
counting theorem.

\begin{theorem}[Stockmeyer~\cite{stockmeyer}]
\label{thm:stockmeyer}
Let $f:\{0,1\}^m \to \{0,1\}$ be polynomial-time computable, and define
\[
p = \Pr_{u \sim \mathcal U(\{0,1\}^m)}[f(u)=1] = \frac{1}{2^m} \sum_{u \in \{0,1\}^m} f(u)\,.
\]
Then for every $\eta,\delta>0$, there exists a randomized algorithm
running in time $\poly(m,1/\eta,\log(1/\delta))$ with access to an
$\NP$ oracle which, with probability at least $1-\delta$, outputs
$\alpha$ such that
\[
(1-\eta)p \le \alpha \le (1+\eta)p.
\]
\end{theorem}

We first prove a general theorem. 

\begin{theorem}
\label{thm:sampling_nonnegative_weights}
Let $g:\{0,1\}^n \to \mathbb R_{\ge 0}$ be polynomial-time computable and
not identically zero. Assume there exist efficiently computable bounds
\[
0<a\le \min_{x:g(x)>0} g(x),
\qquad
\max_{x} g(x)\le A,
\]
such that $A$ and $1/a$ are at most exponential in $n$.
Let
\[
Z := \sum_{x\in\{0,1\}^n} g(x).
\]
Then for every $\epsilon,\delta>0$, there exists a randomized algorithm
running in time $\poly(n,1/\epsilon,\log(1/\delta))$ with access to an
$\NP$ oracle such that, with probability at least $1-\delta$, it samples
from a distribution
\[
q=\{x \text{ w.p. } q(x)\}
\]
satisfying
\[
(1-\epsilon)\frac{g(x)}{Z}
\le q(x)\le
(1+\epsilon)\frac{g(x)}{Z}
\qquad \text{for all } x\in\{0,1\}^n.
\]
\end{theorem}

\begin{proof}
For each prefix $v\in\{0,1\}^{\le n}$, let
\[
S_v:=\{x\in\{0,1\}^n : \text{$x$ has prefix $v$}\},
\qquad
W_v:=\sum_{x\in S_v} g(x).
\]
Then $W_\emptyset=Z$. Moreover, if $W_v>0$, then $W_v\ge a$, since some
$x\in S_v$ has $g(x)>0$ and every positive value of $g$ is at least $a$. We will approximate these prefix weights multiplicatively, and then sample
bit-by-bit using the chain rule.

\paragraph{Step 1: approximating prefix weights.}
Fix a prefix $v$. Let $j$ be an integer to be chosen later, and define a
Boolean function
\[
F_v:\{0,1\}^{n+j}\to\{0,1\}
\]
as follows. On input $(x,t)$, where $x\in\{0,1\}^n$ and $t\in\{0,1\}^j$,
output $0$ if $x\notin S_v$, and otherwise output $1$ if
\[
\frac{t}{2^j} < \frac{g(x)}{A},
\]
and output $0$ otherwise. Since $0\le g(x)/A\le 1$, this predicate is
well-defined and polynomial-time computable.

Let
\[
p_v := \Pr_{(x,t)\sim \mathcal U(\{0,1\}^{n+j})}[F_v(x,t)=1].
\]
For each fixed $x\in S_v$, the acceptance probability over uniformly random
$t$ differs from $g(x)/A$ by at most $2^{-j}$. Therefore
\[
\left| p_v - \frac{W_v}{A\,2^n} \right| \le 2^{-j},
\]
and hence
\[
\left| A\,2^n\,p_v - W_v \right| \le A\,2^{n-j}.
\]
Choose $j$ so that
\[
A\,2^{n-j} \le \eta a,
\]
where $\eta>0$ is a parameter to be fixed later. Since $A/a$ is at most
exponential in $n$, such a choice of $j$ is polynomial in
$n+\log(1/\eta)$.

Applying Theorem~\ref{thm:stockmeyer} to $F_v$, we obtain in
$\poly(n,1/\eta,\log(1/\delta_v))$ time with access to $\NP$ an estimate
$\alpha_v$ such that, with probability at least $1-\delta_v$,
\[
(1-\eta)p_v \le \alpha_v \le (1+\eta)p_v.
\]
Define
\[
\widetilde W_v := A\,2^n\,\alpha_v.
\]
We claim that whenever $W_v>0$, this gives a multiplicative approximation
to $W_v$. Indeed, since $W_v\ge a$ and
$|A\,2^n p_v-W_v|\le \eta a \le \eta W_v$, we have
\[
(1-\eta)W_v - \eta W_v
\le \widetilde W_v
\le (1+\eta)W_v + \eta W_v.
\]
Therefore
\[
(1-2\eta)W_v \le \widetilde W_v \le (1+2\eta)W_v
\]
whenever $W_v>0$.

\paragraph{Step 2: approximate conditional probabilities.}
Suppose we have already sampled a prefix $v$ with $W_v>0$.
The true conditional probability that the next bit equals $1$ is
\[
p_{v\to 1}:=\frac{W_{v1}}{W_v}.
\]
We estimate it by
\[
\widetilde p_{v\to 1}:=\frac{\widetilde W_{v1}}{\widetilde W_{v0} + \widetilde W_{v1}}.
\]
On the event that both $\widetilde W_{v0}$ and $\widetilde W_{v1}$ satisfy the
bounds above, we obtain
\[
\frac{1-2\eta}{1+2\eta}\,\frac{W_{v1}}{W_v}
\le
\widetilde p_{v\to 1}
\le
\frac{1+2\eta}{1-2\eta}\,\frac{W_{v1}}{W_v}.
\]
In particular, for sufficiently small $\eta$,
\[
(1-6\eta)\frac{W_{v1}}{W_v}
\le
\widetilde p_{v\to 1}
\le
(1+6\eta)\frac{W_{v1}}{W_v}.
\]

\paragraph{Step 3: sampling by prefix expansion.}
We now generate a sample $x=x_1x_2\cdots x_n$ iteratively. We start with
the empty prefix $v=\emptyset$. At each step, given the current prefix
$v$, we compute $\widetilde W_v$ and $\widetilde W_{v1}$ and sample the
next bit to be $1$ with probability $\widetilde p_{v\to 1}$ and $0$ with
probability $1-\widetilde p_{v\to 1}$. We then update the prefix and
continue.

If at every step the relevant prefix-weight estimates are correct, then for
every output string $x\in\{0,1\}^n$ we have
\[
q(x)=\prod_{i=1}^n \widetilde p_i(x),
\]
where each $\widetilde p_i(x)$ is the approximate conditional probability
used at step $i$, while the target distribution satisfies
\[
\frac{g(x)}{Z}
=
\prod_{i=1}^n p_i(x),
\]
where each $p_i(x)$ is the corresponding exact conditional probability.
Hence
\[
\left(1-6\eta\right)^n \frac{g(x)}{Z}
\le
q(x)
\le
\left(1+6\eta\right)^n \frac{g(x)}{Z}.
\]
Choosing $\eta= c \cdot\epsilon/n$, for sufficiently small $c > 0$,
gives
\[
(1-\epsilon)\frac{g(x)}{Z}
\le
q(x)
\le
(1+\epsilon)\frac{g(x)}{Z}.
\]

Finally, the algorithm uses at most $2n$ invocations of
Theorem~\ref{thm:stockmeyer} along any sampled path. Setting the failure
probability of each call to $\delta/(2n)$ and taking a union bound shows
that all estimates used during sampling are simultaneously correct with
probability at least $1-\delta$. The total running time remains
$\poly(n,1/\epsilon,\log(1/\delta))$.
\end{proof}

\begin{remark}
The lower bound $a$ is used only to convert the additive discretization
error arising from the threshold test into a multiplicative approximation
for every nonzero prefix weight $W_v$. If the positive values of $g$ could
be doubly exponentially small, this argument would no longer yield a
polynomial-time multiplicative approximation.
\end{remark}

We now apply Theorem~\ref{thm:sampling_nonnegative_weights} to the output
distribution of DQI. A DQI state has the form
\[
\ket{\psi_P} \propto \sum_{x\in\{0,1\}^n} P(f(x))\ket{x},
\]
where $f:\{0,1\}^n\to \mathbb Z$ is polynomial-time computable and
$P$ is a polynomial of degree at most $\poly(n)$. Measuring this state in
the computational basis yields the distribution
\[
\mathcal D_{\mathrm{DQI}}
=
\left\{
x \text{ w.p. } \frac{|P(f(x))|^2}{Z}
\right\},
\qquad
Z=\sum_x |P(f(x))|^2.
\]

Define
\[
g(x):=|P(f(x))|^2.
\]
Then $g$ is polynomial-time computable. Since $f(x)$ is bounded by
$\poly(n)$ and $P$ has degree at most $\poly(n)$ with polynomially many
bits of precision in its coefficients, both $\max_x g(x)$ and the inverse
of the smallest nonzero value of $g(x)$ are at most exponential in $n$.
Therefore Theorem~\ref{thm:sampling_nonnegative_weights} applies.

\begin{corollary}
\label{cor:sample_dqi_bppnp}
For every DQI state $\ket{\psi_P}$ and every $\epsilon,\delta>0$, there
exists a randomized algorithm running in
$\poly(n,1/\epsilon,\log(1/\delta))$ time with access to an $\NP$ oracle
such that, with probability at least $1-\delta$, it samples from a
distribution
\[
q=\{x \text{ w.p. } q(x)\}
\]
satisfying
\[
(1-\epsilon)|\langle x \mid \psi_P\rangle|^2
\le
q(x)
\le
(1+\epsilon)|\langle x \mid \psi_P\rangle|^2
\qquad \text{for all } x\in\{0,1\}^n.
\]
\end{corollary}

The same argument applies to the qudit version of DQI for Max-LINSAT:
one may either work directly over a $d$-ary alphabet or encode each dit in
binary with only polynomial overhead. Corollary~\ref{cor:sample_dqi_bppnp} places approximate classical sampling
from DQI output distributions inside $\BPP^\NP$.

\Cref{thm:sampling_nonnegative_weights} implies a $\BPP^{\NP}$ algorithm to sample the output distribution of \textit{any} quantum state when each amplitude (or its magnitude squared) is efficiently computable. The quantum state does not have to be normalized for this procedure to work.
The output of Shor's discrete log circuit also can be sampled using \Cref{thm:sampling_nonnegative_weights}; we explain this in \Cref{sec:simulate_shor}.

%\pagebreak

\section{DQI implements a 1990 linear programming bound}
\label{sec:semicirclelaw}

\subsection{Overview}
We show that DQI over $\mathbb{F}_2$ admits a precise interpretation in terms of
classical coding theory. In particular, it provides a constructive
realization of a linear programming bound for codes due to
Tietäväinen~\cite{tietavainen}, which itself builds on the
MacWilliams identity~\cite{macwilliams1963theorem}.

The classical argument is existential: it shows that for any target
vector $v$, there exists a distribution over codewords achieving a
desired average distance from $v$, but does not provide an efficient
method to sample from this distribution. We show that the classical argument can be made constructive as a quantum algorithm equivalent to DQI.

The key observation is that the MacWilliams identity~\cite{macwilliams1963theorem} can be interpreted
as a discrete Fourier transform over the hypercube. The polynomials
used in the linear programming bound correspond to low-degree
Fourier combinations of weight-$k$ slices. DQI prepares precisely
such superpositions by coherently summing over dual codewords. As a consequence, DQI can be viewed as a quantum procedure that
constructively realizes the optimal distributions underlying the
linear programming bound. In particular, the
``semicircle law'' observed in~\cite{dqi}, which is the maximum expected number of satisfying assignments of a solution output by DQI,  coincides with
Tietäväinen's asymptotic bound on the covering radius.

This perspective highlights a fundamental distinction: while the
existence of such distributions is classical, their efficient
preparation appears to require coherent Fourier access to the dual
code, which is naturally provided by a quantum algorithm.

DQI works for optimization problems in which one fixes a matrix
$B \in \mathbb{F}_2^{m \times n}$ and a target vector $v \in \mathbb{F}_2^m$,
and scores a bitstring $x \in \mathbb{F}_2^n$ according to the Hamming
distance between $Bx$ and $v$. This admits a natural coding-theoretic interpretation. Define
the linear code
\[
C := \{Bx : x \in \mathbb{F}_2^n\} \subseteq \mathbb{F}_2^m.
\]
Then $v$ may be viewed as a noisy codeword, and the goal is to find a codeword
$c \in C$ that is far from $v$ (or close to $\overline{v}$). The performance of DQI is governed by properties of the dual code
\[
C^\perp = \{y \in \mathbb{F}_2^m : B^T y = 0\}.
\]
In particular, the algorithm requires the ability to decode $C^\perp$
up to a certain radius, which determines the degree of the polynomial
transform that can be implemented.

An identity comparing codewords of $C$ and of $C^\perp$ was proven by Jessie MacWilliams in 1963. The Hamming weight distribution of codewords in $C$ is related to that of $C^\perp$ by the \textit{Kravchuk transform} (see \Cref{defn:krav}). We state this identity as \Cref{fact:macwilliams}.
% :
% \begin{fact}[\cite{macwilliams1963theorem}]
%     Given a binary linear code $C: \mathbb{F}_2^n$ to $\mathbb{F}_2^m$, let 
%     \begin{align*}
%         W_k \defeq \Big| \{c \in C; |c| = k\} \Big| \,, & & W_k^\perp \defeq \Big|\{c \in C^\perp; |c| = k\}\Big|\,.
%     \end{align*}
%     Then the following relationship holds:
%     \begin{align*}
%         W_k^\perp = \frac{1}{|C|} \sum_{i=0}^m W_i K_k(i)\,,
%     \end{align*}
%     where $K_k$ is the $k^{\text{th}}$ Kravchuk polynomial on $m$ inputs (\Cref{defn:krav}).
% \end{fact}
% Note that all $W_k$ and $W_k^\perp$ must be non-negative. 
This gives a method to upper bound the distance of $C$ to any point $v$, called the \textit{covering radius} of $C$.
The key observation is to interpret the Kravchuk transform as a discretized Fourier transform:
\begin{fact}[\cite{tietavainen}]
\label{fact:function_to_bound_coveringradius}
Fix a binary linear code $C$ with dual distance $d^\perp$.
    Let $\beta(x) = \sum_{k=0}^{d^\perp - 1} \beta_k K_k(x)$ be such that  $\beta_0 > 0$ and $\beta(i) \le 0$ for all $i > j$. Then the \textit{covering radius} of $C$ must be at most $j$.
\end{fact}
Notably, \Cref{fact:function_to_bound_coveringradius} is used in \cite{tietavainen} to exactly match the semicircle law of~\cite{dqi} in the limit as $m \to \infty$ as $\frac{\ell}{m}$ held constant.
However, this bound is \emph{existential}, and does not suggest any classical algorithm to find these codewords. 

We first show how to interpret the bound of \cite{tietavainen} as constructing a \emph{good} distribution of codewords $\mathcal{D}_v$; i.e. a codeword randomly sampled from $\mathcal{D}_v$ has a target average distance to $v$.
This naturally implies a \emph{quantum} algorithm, where we construct a quantum state with output distribution $\mathcal{D}_v$. In fact, this quantum algorithm is exactly DQI over $\mathbb{F}_2$.
\begin{theorem}
    Fix a real number $\alpha$ and a binary linear code $C$ with an efficient decoder of $C^\perp$ up to $\ell$ errors. Let $\beta(x) = \sum_{k=1}^{\min(2\ell+1, d^\perp - 1)} \beta_k K_k(x)$ be such that $\frac{\beta(x)}{x-\alpha}$ is a square polynomial.\footnote{A polynomial $p$ is square if $\sqrt{p}$ is also a polynomial.} For each point $v$, there is a quantum algorithm to sample a codeword with average distance $\alpha$ from $v$.
\end{theorem}
This algorithm can locate the subset of codewords at a target distance from $v$ using a good choice of $\beta$. Our best choice of $\beta$ sets the target distance at the semicircle law. Unlike a classical computer, a quantum computer can coherently implement this Fourier transform over the whole codespace. 

As of now, the only way we know how to prove this bound on the covering radius is \emph{through} the Fourier-analytic picture in the MacWilliams identity.
This suggests that efficiently sampling from $\mathcal{D}_v$ requires a coherent Fourier transform, which requires a quantum algorithm.

\subsection{The MacWilliams identity}
The MacWilliams identity relates the Hamming weight distribution of codewords in a code $C$ with that of its dual code $C^\perp$. It relies on the \emph{Kravchuk} polynomials, a set of orthogonal polynomials associated with the binomial distribution (see \cite{feinsilver2007krawtchoukmatricesclassicalquantum} for a survey). These polynomials can be viewed as a discrete version of the Hermite polynomials.
\begin{definition}[\cite{krawtchouk1929generalisation}]
\label{defn:krav}
    For any positive integer $m$, define the Kravchuk (or Krawtchouk) polynomials $\{K_k \ |\ 0 \le k \le m\}$ as
    \begin{align*}
        K_k(j;m) \defeq \sum_{i=0}^k (-1)^i  {j \choose i} {m-j \choose k-i}\,.
    \end{align*}
    When $m$ is clear from context, we write the evaluated polynomial as $K_k(j)$.
\end{definition}
\begin{fact}[{e.g. \cite[Theorem 16]{msbook}}]
\label{fact:krav_orthogonal}
    $\sum_{i=0}^m {m \choose i} K_k(i;m) K_\ell(i;m) = 2^m {m \choose k} \delta_{k\ell}$.
\end{fact}
We now state the MacWilliams identity:
\begin{fact}[\cite{macwilliams1963theorem}]
\label{fact:macwilliams}
    Given a binary linear code $C: \mathbb{F}_2^n \to \mathbb{F}_2^m$, let 
    \begin{align*}
        W_k \defeq \left| \{c \in C; |c| = k\} \right|\,, & & W_k^\perp \defeq \left| \{c \in C^\perp; |c| = k\} \right|\,.
    \end{align*}
    Then the following relationship holds:
    \begin{align*}
        W_k^\perp = \frac{1}{|C|} \sum_{i=0}^m W_i K_k(i)\,.
    \end{align*}
\end{fact}
Note that all $W_k$ and $W_k^\perp$ must be non-negative.
This identity was famously used to upper bound the size of $C$ by its distance~\cite{mceliece1977new}, known as the first and second \emph{linear programming bound} for codes. We will use (a slight generalization of) this identity to upper bound the \emph{covering radius} of a code $C$ by the distance of its \textit{dual} code $C^\perp$.

\subsection{Recovering the semicircle law}
We will re-prove the following existential result of \cite{tietavainen}:
\begin{claim}[\cite{tietavainen}]
\label{claim:tiet}
    The covering radius of a binary linear code $C: \mathbb{F}_2^n \to \mathbb{F}_2^m$ is at most
\begin{align*}
    \frac{m}{2} - \sqrt{\ell(m-\ell)} + o(m)\,,
\end{align*}
for any $\ell$ less than half the dual distance; i.e. $\ell < d^\perp/2$.
\end{claim}

To prove \Cref{claim:tiet}, we need a slight generalization of the MacWilliams identity. This allows us to reason about the distribution of codewords from any point $v \in \mathbb{F}_2^m$.
\begin{fact}
\label{fact:general_macwilliams}
Fix a binary linear code $C: \mathbb{F}_2^n \to \mathbb{F}_2^m$. Let 
$$
W_\ell(v) \defeq \left| \{ c \in C;  |c-v| = \ell\}\right|
$$ be the number of codewords with Hamming distance $\ell$ from a point $v \in \mathbb{F}_2^m$. Then
\begin{align}
\label{eqn:mcw_general}
    \sum_{i=0}^m  W_i(v) K_k(i) = \sum_{y \in \mathbb{F}_2^m; |y|=k}(-1)^{v \cdot y} \cdot |C| \cdot \mathbf{1}_{B^T y = 0}\,.
\end{align}

\end{fact}
\begin{proof}
We rely on the following identity, which holds for any $w \in \mathbb{F}_2^m$:
$$
K_k(|w|) = K_k(\frac{m - \sum_{i=1}^m (-1)^{w_i}}{2}) 
% = P^{(k)}((-1)^{w_1}, \dots, (-1)^{w_m}) 
= \sum_{y \in \mathbb{F}_2^m; |y|=k} (-1)^{w \cdot y} 
\,.
$$
We prove this identity as \Cref{fact:count_elementary_sym_polynomials} in \Cref{sec:mathfacts}.

Let $B \in \mathbb{F}_2^{m \times n}$ be the generating matrix of $C$. 
% Note that $\sum_{\ell = 0}^m W_\ell(v) = |C|$. 
Then
$$
\sum_{i=0}^m  W_i(v) K_k(i) = \sum_{c \in C} K_k(|c-v|) 
= \sum_{c \in C} \sum_{y \in \mathbb{F}_2^m; |y|=k} (-1)^{(c-v) \cdot y}
 = \sum_{y \in \mathbb{F}_2^m; |y|=k}  (-1)^{v \cdot y} \sum_{x \in \mathbb{F}_2^n} (-1)^{(B x) \cdot y}\,.
 $$
Observe that the second sum is nonzero only when $Bx \cdot y = 0$ for all $x$; i.e. when $B^T y = 0$.
\end{proof}
When $v = 0$,  the right-hand side of \cref{eqn:mcw_general} counts the number of \emph{dual} codewords in $C^\perp$ of weight $k$; i.e. $|C| \cdot W_k^\perp(0)$. This recovers \Cref{fact:macwilliams}.

We now use \Cref{fact:general_macwilliams} to bound the covering radius:
\begin{claim}[\cite{tietavainen}]
\label{claim:function_to_bound_coveringradius}
Fix a binary linear code $C$ with dual distance $d^\perp$.
    Let $\beta(x) = \sum_{k=0}^{d^\perp - 1} \beta_k K_k(x)$ be such that  $\beta_0 > 0$ and $\beta(i) \le 0$ for all $i > j$. Then the covering radius of $C$ must be at most $j$.
\end{claim}
\begin{proof}
Suppose $C$ is generated by some matrix $B \in \mathbb{F}_2^{m \times n}$.  The \emph{dual code} $C^\perp$ is then equal to  $\{y \ |\  B^T y = 0\}$. Fix any point $v \in \mathbb{F}_2^m$. Invoking \Cref{fact:general_macwilliams} for all $1 \le k < d^\perp$,  we have
\begin{align*}
    \sum_{i=0}^m W_i(v) K_k(i)= 0\,.
    % = \sum_{c \in C} K_k(|c-v|) = 0\,.
\end{align*}
Since $K_0(x) = 1$ is constant, we have $\sum_{i=0}^m W_i(v) K_0(i) = |C|$.

We can trivially represent $\beta_0 = \beta_0 \cdot 1 + \sum_{k=1}^{d^\perp - 1} \beta_k \cdot 0$. So then
$$
\beta_0 = \sum_{k=0}^{d^\perp - 1} \beta_k \cdot \frac{1}{|C|} \sum_{i=0}^m W_i(v) K_k(i)\,,
$$
since the second sum is $1$ at $k = 0$ and $0$ for all $1\le k < d^\perp$. By definition of $\beta$, we rewrite this as
\begin{align}
\label{eqn:beta0_simplified}
    \beta_0 =  \frac{1}{|C|} \sum_{i=0}^m W_i(v) \beta(i)\,.
\end{align}
Recall that $\beta_0 > 0$ and $\beta(i) < 0$ for all $i > j$ by assumption. Since  $W_i(v) \ge 0$, we have
\begin{align}
\label{eqn:bounding_betazero}
    0 < \beta_0 = \frac{1}{|C|} \sum_{i=0}^m W_i(v) \beta(i) \le \frac{1}{|C|} \sum_{i=0}^{j} W_i(v) \beta(i)\,.
\end{align}
So, for any $v$, there is some $W_i(v) > 0$ for $0 \le i \le j$. So the covering radius of $C$ is at most $j$.
\end{proof}
\Cref{claim:function_to_bound_coveringradius} means that a good choice of $\beta(x)$ with degree $< d^\perp$ implies a good bound on the \emph{covering radius} of $C$. We first describe the choice proposed by \cite{tietavainen}:
\begin{fact}[\cite{tietavainen}]
\label{fact:tietavainen_fn}
For any
$\ell$, let $\alpha_1 \le \dots \le  \alpha_\ell$ be the roots of $K_\ell(x)$. Then for any $\epsilon > 0$,
\begin{align*}
    \beta(x) = -(x-\alpha_1-\epsilon) \prod_{s=2}^\ell (x-\alpha_s)^2
\end{align*}
has $\beta_0 > 0$, and $\beta(i) \le 0$ for all $i > \alpha_1 + \epsilon$.  
\end{fact}
\begin{proof}
When $i > \alpha_1 + \epsilon$, we have $-(x-\alpha_1-\epsilon) < 0$, and so $\beta(x)$ cannot be positive.

Now we prove $\beta_0 > 0$. We rewrite $\beta_0$ using \Cref{fact:krav_orthogonal}:
$$
\sum_{i=0}^m {m \choose i} \beta(i)
=
 \sum_{i=0}^m {m \choose i} K_0(i) \beta(i)
 =
 \sum_{k=0}^{d^\perp - 1} \beta_k \sum_{i=0}^m {m \choose i} K_0(i) K_k(i)
 =
 \beta_0 \cdot 2^m 
\,.
$$
When $\epsilon = 0$, we have
$$
\beta_0 \cdot 2^m = \sum_{i=0}^m {m \choose i} K_\ell(i) \cdot \frac{K_\ell(i)}{x - \alpha_1}\,.
$$
Again by \Cref{fact:krav_orthogonal}, this value is zero, since the second term has degree strictly less than $\ell$.

We take a derivative of $\beta_0$ with respect to $\epsilon$. This value is
\begin{align*}
     \frac{1}{2^m} \sum_{i=0}^m {m \choose i} \prod_{s=2}^\ell (x-\alpha_s)^2 > 0\,,
\end{align*}
and so $\beta_0 > 0$ when $\epsilon$ is positive.
\end{proof}

We may now prove the main result of \cite{tietavainen}:
\begin{proof}[Proof of \Cref{claim:tiet}]
By \Cref{claim:function_to_bound_coveringradius} and \Cref{fact:tietavainen_fn}, the covering radius of a code $C$ is at most the smallest root of $K_\ell(x;m)$, for $2 \ell + 1 < d^\perp$.
The asymptotic value of this root is well-known  (e.g., \cite[Theorem 35]{slot2023sum}). If $m \to \infty$ while $\frac{\ell}{m}$ stays constant, the value $\alpha_1$ approaches
\begin{align*}
\frac{m}{2} - \sqrt{\ell(m-\ell)}\,.\tag*{\qedhere}
\end{align*}   
\end{proof}

\subsection{Interpreting this bound as a quantum algorithm}

\Cref{claim:tiet} is precisely the ``semicircle law'' of DQI~\cite{dqi}.
However, the  bound is existential, whereas DQI actually \emph{locates} a codeword at this distance. Here we show how to interpret the proof of \Cref{claim:tiet} as a quantum algorithm.

Let $\mathcal{D}$ be a distribution over codewords in $C$, where $p: C \to \mathbb{R}_{\ge 0}$ computes the probability of observing a codeword in $\mathcal{D}$. For any vector $v \in \mathbb{F}_2^m$, we write $\E_{\mathcal{D}}[|c-v|]$ as the expected distance from $v$ of a codeword sampled from $\mathcal{D}$. Then the covering radius is at most
$$
\max_v \min_{\mathcal{D}} \E_\mathcal{D}[|c-v|]\,.
$$
For each $v$, we consider $\mathcal{D}_v$ where $p_v(c)$ is a function of $|c-v|$. Intuitively, we will construct $\mathcal{D}_v$ by placing high probability on codewords that are close to $v$.
\begin{theorem}
\label{thm:tiet_firstmoment}
    Fix a real number $\alpha$ and a binary linear code $C$ with dual distance $d^\perp$. Let $\beta(x) = \sum_{k=0}^{d^\perp - 1} \beta_k K_k(x)$ be such that $\beta_0 \le 0$ and $\frac{\beta(x)}{x-\alpha}$ is a square polynomial. Then the covering radius of $C$ is at most $\alpha$.
\end{theorem}
\begin{proof}
Let $\gamma^2(x) \defeq \frac{\beta(x)}{x-\alpha}$; then $\gamma$ is a polynomial of degree at most $(d^\perp-1)/2 $.
Let $\mathcal{D}_v$ be the distribution where $p_v(c) = \frac{1}{Z} \cdot \gamma^2(|c-v|)$ and $Z = \sum_{c \in C} \gamma^2(|c-v|)$. By the first moment method, the covering radius is at most
\begin{align*}
    \max_v \E_{\mathcal{D}_v}[|c-v|] 
    &= \sum_{c \in C} p_v(c) |c-v| 
    \\
    &= \alpha + \sum_{c \in C} p_v(c) (|c-v|-\alpha)
    \\
    &= \alpha + \frac{1}{Z} \sum_{c \in C} \beta(|c-v|)\,.
    \end{align*}
The second term is $\frac{1}{Z} \cdot |C| \cdot \beta_0$ by \cref{eqn:beta0_simplified}, and so negative by assumption. Then the covering radius is at most $\alpha$.
\end{proof}
We can bound the covering radius with essentially the same polynomial as in \Cref{thm:tiet_firstmoment}. This again implies \Cref{claim:tiet}:
\begin{claim}
\label{claim:tiet_by_distribution}
    Fix a binary linear code $C$ with dual distance $d^\perp$. Then the covering radius of $C$ is at most the smallest root of $K_{\ell}$ for any $2 \ell + 1 < d^\perp$.
\end{claim}
\begin{proof}
Fix any valid choice of $\ell$. Let $\alpha_1 \le  \dots \le \alpha_\ell$ be the roots of $K_\ell(x)$. Consider the polynomial
\begin{align*}
    \beta(x) = (x-\alpha_1) \prod_{s=2}^\ell (x-\alpha_s)^2\,.
\end{align*}
Following the proof of \Cref{fact:tietavainen_fn}, we have $\beta_0 = 0$. Moreover, $\frac{\beta(x)}{x-\alpha_1}$ is a square polynomial by construction. The claim follows by \Cref{thm:tiet_firstmoment}.
\end{proof}
How could we obtain a nearby codeword to the vector $v$?
\Cref{thm:tiet_firstmoment} suggests a natural approach: Sample a codeword from $\mathcal{D}_v$. It is not clear how to classically sample from the distribution $\mathcal{D}_v$. However, we can do so quantumly, if we can prepare the state
\begin{align}
\label{eqn:tiet_dist}
    \sum_{x \in\mathbb{F}_2^n} \sqrt{p_v(Bx)}\ket{x} = \frac{1}{\sqrt{Z}}  \sum_{x \in\mathbb{F}_2^n} \gamma(|Bx-v|)\ket{x}
    \,.
\end{align}
Distributions allowed by \Cref{thm:tiet_firstmoment} are exactly the ones constructible with DQI.
\begin{theorem}
\label{claim:howdqiworks}
    Fix a real number $\alpha$ and a binary linear code $C$ with an efficient decoder of $C^\perp$ up to $\ell$ errors. Let $\beta(x) = \sum_{k=1}^{\min(2\ell+1, d^\perp - 1)} \beta_k K_k(x)$ be such that $\frac{\beta(x)}{x-\alpha}$ is a square polynomial. For each point $v$, there is a quantum algorithm that produces a codeword with average distance $\alpha$ from $v$.
\end{theorem}
\begin{proof}
Let $\gamma^2(x) = \frac{\beta(x)}{x-\alpha}$, and $\gamma(x) = \sum_{k=0}^{\ell} \gamma_k K_k(x)$. By construction, $\beta_0 = 0$. So the claim holds by following the proof of \Cref{thm:tiet_firstmoment}, if we can prepare
\begin{align*}
    \frac{1}{\sqrt{Z}}  \sum_{x \in\mathbb{F}_2^n} \gamma(|Bx-v|)\ket{x}
    =
     \frac{1}{\sqrt{Z}}  \sum_{k=0}^{\ell} \gamma_k \sum_{x \in\mathbb{F}_2^n} K_k(|Bx-v|)\ket{x}\,.
\end{align*}
We rely on the identity $K_k(|w|) =  \sum_{y \in \mathbb{F}_2^m; |y|=k} (-1)^{w \cdot y}$ (proven as \Cref{fact:count_elementary_sym_polynomials} in \Cref{sec:mathfacts}). The Fourier transform of the quantum state is proportional to
\begin{align}
\label{eqn:dqistate_with_ft}
   \sum_{k=0}^{\ell} \gamma_k \cdot \sum_{y \in \mathbb{F}_2^m; |y|=k} (-1)^{v \cdot y} \ket{B^T y}\,.
\end{align}
States of the form (\ref{eqn:dqistate_with_ft}) are exactly preparable with DQI; i.e. as \cite[Equation 3]{dqi}.  
\end{proof}
We may interpret \Cref{claim:howdqiworks} as both a constructive variant of \Cref{thm:tiet_firstmoment} and a performance analysis of the DQI algorithm. 
The quantum algorithm takes advantage of information about the dual code $C^\perp$ to construct states $\sum_{|y|=k} (-1)^{v \cdot y} \ket{B^T y}$ in the Fourier basis, and then prepares $\mathcal{D}_v$ by coherently assigning probabilities to codewords by distance to $v$.
Meanwhile, there is no obvious way to \textit{classically} construct the distributions $\mathcal{D}_v$ in \Cref{claim:howdqiworks}.

Using the choice of $\beta$ in \Cref{claim:tiet_by_distribution}, a quantum algorithm can obtain codewords at semicircle law distance, or in fact at any distance in the region $[\frac{m}{2} - \sqrt{\ell(m-\ell)} + o(m), \frac{m}{2} + \sqrt{\ell(m-\ell)} + o(m)]$:
\begin{corollary}
\label{cor:dqi_on_any_kravchuk_root}
    Consider a binary linear code with an efficient decoder of $C^\perp$ up to $\ell$ errors. Let $\alpha$ be any root of a Kravchuk polynomial $K_k$ for $k \le \ell$. Then a quantum algorithm can prepare a distribution over codewords of average distance $\alpha$.
\end{corollary}
\begin{proof}
Let $\alpha_1 \le \dots \le \alpha_k$ be the roots of $K_k(x)$. Consider the polynomial
\begin{align*}
    \beta(x) = (x-\alpha_i) \prod_{s=1; s \ne i}^k (x-\alpha_s)^2\,.
\end{align*}
Following the proof of \Cref{fact:tietavainen_fn}, we have $\beta_0 = 0$. By construction, $\beta$ is a polynomial of degree $2k-1 \le 2\ell-1 \le d^\perp$, and $\frac{\beta(x)}{x-\alpha_i}$ is a square polynomial. The claim follows by \Cref{claim:howdqiworks}.
\end{proof}
The \textit{existence} of codewords everywhere in this region can also be proven using a minor tweak of \cite{tietavainen}'s original work; see \Cref{sec:codewords_concentric_hamming_balls}. A quantum algorithm can locate these codewords.
For example, using \Cref{cor:dqi_on_any_kravchuk_root} with the smallest root of $K_k$, we may quantumly find codewords of average asymptotic distance 
$$
\frac{m}{2} - \sqrt{k(m-k)} + o(m)\,,
$$
for any $0 \le k \le \ell$. It is unclear if a classical algorithm can find these codewords.

%\pagebreak

\section{DQI and the quantum harmonic oscillator}
 \label{sec:harmonic-oscillator}
\subsection{Overview}

We show how DQI on $n$ variables and $m$ constraints can be seen as acting on a quantum harmonic oscillator with $\poly(m,n)$ qubits. However, even if the quantum harmonic oscillator can be analyzed classically, actually \textit{preparing} a state requires a quantum computer. This allows a quantum computer to access information stored in the states of the oscillator.
By analogy, the oracle quantum algorithm traversing an exponentially large welded tree~\cite{childs2003exponential} can be seen as acting on a polynomially-sized line, yet there is no efficient classical algorithm to traverse the tree. 

In the usual formulation of a quantum harmonic oscillator, a quantum particle lying on the real line is subject to a quadratic potential well. This system can be exactly solved: The eigenstates have energy $(k + \frac{1}{2})$ for $k \in \mathbb{Z}_{\ge 0}$, and in the position basis, are Hermite polynomials with a Gaussian decay. 
However, digital quantum computers have a finite number of basis states, and cannot act on the real line. 
We choose a procedure to discretize the oscillator that preserves the eigenspectrum~\cite{atakishiev1990difference}.  We then implement this oscillator on a quantum computer so that the eigenstates are Fourier transforms of $n$-qubit Dicke states.

\begin{claim}
    Fix any positive integer $s$. There is a Hamiltonian $\mathbf{H}_{s}$ that specifies
    a \textit{discrete} quantum harmonic oscillator $\mathbf{H}_{s}$ over the grid $\{-s, \dots, 0, \dots, s\}$ in $m \defeq 2s$ qubits, with position operator 
    \begin{align*}
    % \label{eqn:positionoperator}
            \hat{x} \defeq -\frac{1}{2} \sum_{i=1}^{m} Z_i\,,
    \end{align*}
    and where the eigenstates of $\mathbf{H}_{s}$ are Fourier transforms of Dicke states; i.e. for $0 \le k \le m$,
    \begin{align*}
        \ket{\lambda_k} = H_2^{\otimes m} \ket{D_k}\,.
    \end{align*}
\end{claim}
In this setup, one can quantumly prepare any state of the discrete harmonic oscillator first by preparing a linear combination of Dicke states, and then applying the Fourier transform. 

We also consider discrete oscillators where the eigenstates are \textit{obfuscated} by some function $h$; i.e. replacing the Dicke state $\sum_{|x|=k} \ket{x}$ with $\sum_{|x|=k} \ket{h(x)}$. This is possible over subspaces of the oscillator where the map $h$ is invertible. In the setting of a linear code, i.e. given $C^\perp = \{y\ |\  B^T y = 0\}$, let $h$ be the map from a noisy codeword $y$ to a syndrome $B^T y$. Then this map is invertible when $|y| < d^\perp / 2$.  These obfuscated states form a low-energy subspace of a quantum harmonic oscillator:
\begin{claim}[informal]
    Consider a linear code with parity check matrix $B^T: \mathbb{F}_2^m \to \mathbb{F}_2^n$ and distance $d$. Then it is possible to embed the first $\lfloor (d^\perp-1)/2\rfloor$ eigenstates of a discrete quantum harmonic oscillator in a quantum computer, where each eigenstate has form $\ket{\lambda_k} := H_2^{\otimes m} \sum_{y \in \mathbb{F}_2^m; |y|=k} \ket{B^T y}$, and with approximate position operator
    $$
    \hat{x}  = -\frac{1}{2} \sum_{x \in \mathbb{F}_2^n} \sum_{i=1}^m (-1)^{(Bx)_i} \ketbra{x}\,.
    $$
\end{claim}
We visualize the accessible states of the obfuscated oscillator in \Cref{fig:qho}.

\begin{figure}[t]
  \centering
  \resizebox{\textwidth}{!}{%
    \input{qho_v40_selfcontained.tex}%
  }
  \caption{\small Cartoon of an obfuscated quantum harmonic oscillator. On the left, we display eigenstates of the oscillator $\ket{\lambda_k}$ in the position basis. On the right, we display an example state which is a linear combination of $\ket{\lambda_k}$ for $k \le 3$.  Only eigenstates with large energy put significant amplitude far from the origin. This limits the maximum (and minimum) position available to states composed of low-energy eigenstates. 
  }
  \label{fig:qho}
\end{figure}

From here, it is easy to interpret DQI as preparing a state in this low-energy subspace of a quantum harmonic oscillator. When $\vec{v} = 1$, the position operator $\hat{x}$ is $m/2$ less than the satisfying fraction of an assignment, so DQI will prepare an oscillator state with as \textit{large} a position as possible.
States in this low-energy subspace take position close to zero; the largest expected position is exactly the semicircle law of~\cite{dqi}.

We explain how the DQI algorithm can be seen as acting on the low-energy subspace of a quantum harmonic oscillator. We show how to discretize the harmonic oscillator in a way that preserves the spectrum, and implement it on a quantum computer. 
We then show how  to encode information in the states of this system that is accessible by sampling the state but not by any obvious classical algorithm. We discuss how this perspective suggests physically motivated generalizations of DQI.

\subsection{A discretized quantum harmonic oscillator}
A quantum harmonic oscillator describes a quantum particle in a quadratic potential well. It is one of the few physical models that is exactly solvable. Consider the differential equation
$$
-i \frac{\partial}{\partial t} \Psi(x;t) = \mathbf{H} \cdot \Psi(x;t)\,,
$$
where 
$$
\mathbf{H} = \frac{1}{2} \left( - \frac{\partial^2}{\partial x^2} + x^2 \right)\,.
$$
The equation has stationary solutions $\Psi_k(x;t) = e^{i (k+\frac{1}{2})t} \cdot \Psi_k(x)$ for each $k \in \mathbb{Z}_{\ge 0}$, where
$$
\Psi_k(x) = \frac{1}{\sqrt{\sqrt{\pi} 2^k k!}} H_k(x) e^{-x^2/2}\,,
$$
and $H_k(x)$ is the $k^{\text{th}}$ \emph{Hermite polynomial}.\footnote{The Hermite polynomials are defined by the relation $H_{-1} = 0$, $H_0 = 1$, and $H_{k+1}(x) = 2x H_k(x) - 2k H_{k-1}(x)$.}

We discretize this system in a way that preserves the spectrum~\cite{atakishiev1990difference,hakioglu2000canonical,chauleur2024discrete}.
The Hermite polynomials are orthogonal with respect to the Gaussian distribution, and complete for square integrable functions on the real line. If we restrict to a grid of $m$ points, then a natural substitute are the \emph{Kravchuk} polynomials (\Cref{defn:krav}) $\{K_0, \dots, K_m\}$. These polynomials are orthogonal with respect to the \textit{binomial} distribution, and approximate the Hermite polynomials in the continuum limit $m \to \infty$ (e.g. \cite{hakioglu2000canonical,chauleur2024discrete}). From here, we may construct a Hamiltonian that is diagonal in the Kravchuk basis:
\begin{fact}[\cite{atakishiev1990difference,chauleur2024discrete}]
\label{fact:discrete_qho}
Fix any positive integer $s$. 
Let $\mathbf{H}_s$ act on the grid $\{-s, \dots, s\}$ defined by the following equation for each integer $0 \le k \le 2s$:
    $$
    \mathbf{H}_s\cdot  K_{k}(x+s;2s)  \sqrt{ {2s \choose x+s}} = (k + \frac{1}{2})\cdot  K_{k}(x+s;2s) \sqrt{ {2s \choose x+s}}\,.
    $$
    Then $\mathbf{H}_s$ approximates a quantum harmonic oscillator in the continuum limit $s \to \infty$.
\end{fact}
Using \Cref{fact:discrete_qho}, we may embed this oscillator on a quantum computer of $2s$ qubits. 
The other benefit of this implementation is that the eigenstates of the oscillator are easy to describe.
\begin{claim}
\label{claim:qho_on_computer}
    Fix any positive integer $s$. Then we can implement a \textit{discrete} quantum harmonic oscillator $\mathbf{H}_{s}$ over the grid $\{-s, \dots, 0, \dots, s\}$ in $m \defeq 2s$ qubits, with position operator 
    \begin{align}
    \label{eqn:positionoperator}
            \hat{x} \defeq -\frac{1}{2} \sum_{i=1}^{m} Z_i\,,
    \end{align}
    and where the eigenstates of $\mathbf{H}_{s}$ are Fourier transforms of Dicke states; i.e. for $0 \le k \le m$,
    \begin{align*}
        \ket{\lambda_k} = H_2^{\otimes m} \ket{D_k}\,.
    \end{align*}
\end{claim}
\begin{proof}
We start from the discrete oscillator in \Cref{fact:discrete_qho}. We represent position $-s \le x \le s$ with the $(x+s)^{\text{th}}$ \textit{Dicke} state
$$
\ket{D_{x+s}} \defeq \frac{1}{\sqrt{ {2s \choose x+s}}} \sum_{y \in \mathbb{F}_2^{2s}; |y|=x+s}\ket{y}\,.
$$
The reader can verify that Dicke states are eigenstates of the position operator in \cref{eqn:positionoperator} with the correct eigenvalue; i.e. $\ket{D_{x+s}}$ has eigenvalue $x$.

We now study the eigenstates of $\mathbf{H}_s$. The $k^{\text{th}}$ eigenstate $\ket{\lambda_k}$ of the oscillator has the form 
$$
\ket{\lambda_k} \propto \sum_{x=-s}^s K_{k}(x+s) \sqrt{ {2s \choose x+s}  } \ket{D_{x+s}} 
= \sum_{x=-s}^s  \sum_{y \in \mathbb{F}_2^{2s}; |y|=x+s} K_k(x+s) \ket{y}
=  \sum_{y \in \mathbb{F}_2^{2s}} K_k(|y|) \ket{y}\,.
$$
We invoke the following identity, proven as  \Cref{fact:count_elementary_sym_polynomials} in \Cref{sec:mathfacts}:
$$
K_k(|w|) = K_k(\frac{m - \sum_{i=1}^m (-1)^{w_i}}{2}) 
% = P^{(k)}((-1)^{w_1}, \dots, (-1)^{w_m}) 
= \sum_{z \in \mathbb{F}_2^m; |z|=k} (-1)^{w \cdot z} 
\,.
$$
Then the eigenstates of $\mathbf{H}_s$ are
\begin{align*}
    H_2^{\otimes m} \ket{D_k} = \frac{1}{\sqrt{{m \choose k}}}  \sum_{z \in \mathbb{F}_2^m; |z|=k} H_2^{\otimes m} \ket{z}
    =
    \frac{1}{\sqrt{2^m {m \choose k}}}  \sum_{z \in \mathbb{F}_2^m; |z|=k} 
    \sum_{y \in \mathbb{F}_2^m} (-1)^{y \cdot z} \ket{y} = \ket{\lambda_k}\,.\tag*{\qedhere}
\end{align*}
\end{proof}

Since the oscillator is polynomially-sized in \Cref{claim:qho_on_computer}, it is easy to implement this in both classical and quantum computers. However, this no longer trivially holds when the oscillator states encode information, as we explore in the next subsection.

\subsection{Obfuscating the harmonic oscillator}
There is no a priori reason that an eigenstate of our harmonic oscillator must be the Fourier transform of a Dicke state. Consider a reversible map $\pi: \mathbb{F}_2^m \to \mathbb{F}_2^m$. We could imagine implementing our oscillator so that the $k^{\text{th}}$ eigenstate is 
\begin{align}
    \label{eqn:oscillator}
    H_2^{\otimes m} \frac{1}{\sqrt{{ m \choose k}}} \sum_{z \in \mathbb{F}_2^m; |z|=k} \ket{\pi(z)} \propto \sum_{z \in \mathbb{F}_2^m; |z|=k} \sum_{y \in \mathbb{F}_2^m} (-1)^{y \cdot \pi(z)} \ket{y}\,.
\end{align}
This modifies the position operator. Suppose there exists some $\sigma: \mathbb{F}_2^m \to \mathbb{F}_2^m$ such that for all $y,z \in \mathbb{F}_2^m$, we have $\sigma(y) \cdot z = y \cdot \pi(z)$. Then the eigenstate is proportional to
$$
\sum_{y \in \mathbb{F}_2^m} K_k(|\sigma(y)|)\ket{y}\,,
$$
which transforms each position state $\ket{D_{x+s}}$ to
$$
\frac{1}{\sqrt{{ m \choose x + s}}} \sum_{y \in \mathbb{F}_2^m; |y| = x+ s} \ket{\sigma^{-1}(y)}\,.
$$
In other words, position $-m/2 \le x \le m/2$ is represented by the uniform superposition over strings $y \in \mathbb{F}_2^m$ with $|\sigma(y)| = x+m/2$.

We may also consider a hash function $\pi: \mathbb{F}_2^m \to \mathbb{F}_2^n$ for some $n < m$. We can again build a discrete harmonic oscillator state whenever it is supported on inputs to $\pi$ that have no collisions. A natural way to do this is with a linear code:
\begin{claim}
\label{claim:obfuscated_qho}
    Fix a positive even $m$ and smaller positive integer $n$. Consider a linear code with parity check matrix $B^T: \mathbb{F}_2^m \to \mathbb{F}_2^n$ and distance $d^\perp$. Then it is possible to embed the first $\lfloor (d^\perp-1)/2\rfloor$ eigenstates of a discrete quantum harmonic oscillator with the form \cref{eqn:oscillator}. 
\end{claim}
\begin{proof}
We again start from $\mathbf{H}_{m/2}$ in \Cref{fact:discrete_qho}. Set $\pi = B^T$. We use the form of the oscillator eigenstate from \cref{eqn:oscillator}:
$$
\ket{\lambda_k} \defeq   H_2^{\otimes n} \frac{1}{\sqrt{{ m \choose k}}} \sum_{z \in \mathbb{F}_2^m; |z|=k} \ket{\pi(z)} = \frac{1}{\sqrt{2^n { m \choose k}}} \sum_{z \in \mathbb{F}_2^m; |z|=k} \sum_{y \in \mathbb{F}_2^n} (-1)^{y \cdot B^T z} \ket{y}
  \,.
$$
Note that $\sigma = B$ has the property that $By \cdot z = y \cdot B^T z$. So this state can be written as
$$
 \frac{1}{\sqrt{2^n { m \choose k}}} \sum_{y \in \mathbb{F}_2^n} K_k(|By|) \ket{y}
  \,.
$$
These states are orthogonal whenever the states $\sum_{z \in \mathbb{F}_2^m; |z| = k} \ket{\pi(z)}$ are orthogonal. This holds for all $0 \le k < d^\perp / 2$, since in this region, $\pi(z) = B^T z$ can be uniquely decoded to $z$.
We then embed $\mathbf{H}_{m/2}$ so that the low-energy eigenstates are $\{\ket{\lambda_0}, \dots, \ket{\lambda_{\lfloor (d^\perp-1)/2\rfloor}}\}$; the higher-energy eigenstates can be chosen arbitrarily. 
\end{proof}
In \Cref{claim:obfuscated_qho}, we again embed low-energy states of a discrete harmonic oscillator in a quantum computer. 
This time, the states themselves encode information: the ``position'' basis approximately encodes the satisfying fraction of the problem $Bx = \vec{1}$.  As a result, \textit{sampling} from states in this harmonic oscillator allows one to find decent approximate solutions to this problem. We explore this idea throughout the rest of the section.

First suppose $\pi$ was invertible. Then \Cref{claim:obfuscated_qho} represents position $-m/2 \le x \le m/2$ with the uniform superposition over bitstrings $w \in \mathbb{F}_2^m$ where $|Bw| = x + m/2$:
$$
\ket{D'_{x+m/2}} \propto \sum_{y \in \mathbb{F}_2^m; |By| = x + m/2} \ket{y}\,.
$$
The reader can verify that $\ket{D'_{x+m/2}}$ has eigenvalue $x$ with the following operator:
$$
\hat{x} \defeq -\frac{1}{2} \sum_{x \in \mathbb{F}_2^n} \sum_{i=1}^m (-1)^{(B x)_i} \ketbra{x}\,.
$$
In our case, $\pi$ is only invertible when the input has Hamming weight below $d/2$. However, the effect of $\hat{x}$ on our eigenstates is still easy to describe. In fact, it matches the effect of a $\hat{J}_x$ operator on angular momentum eigenstates:
\begin{lemma}
\label{lemma:position_operator_is_bosonic}
    Let $\{\ket{\lambda_i}\}$ be the harmonic oscillator states from \Cref{claim:obfuscated_qho} with parity check matrix $B^T: \mathbb{F}_2^m \to \mathbb{F}_2^n$ and distance $d^\perp$. Let $c_\ell \defeq \sqrt{\ell(m-\ell+1)}$. Then the operator
       $$
    \hat{x} \defeq -\frac{1}{2} \sum_{x \in \mathbb{F}_2^n} \sum_{i=1}^m (-1)^{(B x)_i} \ketbra{x}\,,
    $$ 
    satisfies  $-2\hat{x} \ket{\lambda_k} = c_k \ket{\lambda_{k-1}} + c_{k+1} \ket{\lambda_{k+1}}$ for all $k < d^\perp/2 - 1$.
\end{lemma}
 We prove \Cref{lemma:position_operator_is_bosonic} in \Cref{sec:mathfacts}. 

If we can prepare the oscillator eigenstates $\{\ket{\lambda_i}\}$, we can make a linear combination $\ket{\psi}$ with large expected position $\hat{x}$. Since $\bra{\psi} \hat{x}\ket{\psi} = |Bx| - \frac{m}{2}$, sampling from this state will give outcomes with large value of $|Bx|$.
This is exactly what the DQI algorithm does~\cite{dqi}. We prepare the Dicke state $\ket{D_k}$, apply $B^T$ in superposition, and uncompute the first register, to prepare the Fourier transform of $\ket{\lambda_k}$:
$$
\frac{1}{\sqrt{ {m \choose k}}} \sum_{z \in \mathbb{F}_2^m; |z|=k} \ket{B^T z}\,.
$$
This is possible whenever we can efficiently decode $k$ errors from the code with parity matrix $B^T$. 
DQI then takes a linear combination of preparable oscillator eigenstates with a large value of $\hat{x}$.

The harmonic oscillator perspective of \Cref{lemma:position_operator_is_bosonic} makes it easy to see what kinds of DQI states are possible. Assuming we can decode $\ell$ errors, we can create states concentrated at essentially any position $\hat{x}$ in $[-\sqrt{\ell(m-\ell)}, \sqrt{\ell(m-\ell)}]$. This can be done by preparing states that are near-eigenvectors of $\hat{x}$ in the $m \to \infty$ limit. See \Cref{fig:codewords_with_dqi} for a cartoon.
\begin{claim}
\label{claim:concentrated_qho_state}
Let $\{\ket{\lambda_i}\}$ be the harmonic oscillator states from \Cref{claim:obfuscated_qho} with parity check matrix $B^T: \mathbb{F}_2^m \to \mathbb{F}_2^n$ and distance $d^\perp$. Consider the state
    $$
    \ket{\psi} = \frac{1}{\sqrt{r-1}} \sum_{i=1}^{r-1} \ket{\lambda_{\ell+i}}\,,
    $$
    where $\ell  + r < d^\perp/2 - 1$.
    We take the limit $m \to \infty$, where $\ell / m$ is held constant, and $\omega(1) < r < o(\ell)$. Then $\ket{\psi}$ has $1 - o(1)$ support on eigenvectors of
    $$
    \hat{x} \defeq -\frac{1}{2} \sum_{x \in \mathbb{F}_2^n} \sum_{i=1}^m (-1)^{(B x)_i} \ketbra{x}
    $$
   with eigenvalue $(1 \pm o(1)) \cdot \sqrt{\ell(m-\ell)}$.
\end{claim}
\begin{proof}
We rely on \Cref{lemma:position_operator_is_bosonic}. Recall that $c_\ell \defeq \sqrt{\ell(m-\ell+1)}$. For any $1 \le i \le r$,
$$
\left| \frac{c_{\ell+1} - c_{\ell+i}}{c_{\ell+1}} \right|
\le \left| \frac{c_{\ell+1} - c_{\ell+r+1}}{c_{\ell+1}} \right|
= \left| 1 - \sqrt{\frac{(\ell+r+1)(m-\ell-r)}{{(\ell+1)(m-\ell)}}} \right| = o(1)\,.
$$
So then
\begin{align*}
2\hat{x} \ket{\lambda_{\ell+i}} = 
c_{\ell+i} \ket{\lambda_{\ell+i-1}} 
+ c_{\ell+i+1} \ket{\lambda_{\ell+i+1}} 
\end{align*}
is $1 - o(1)$ close to $c_{\ell+1} \left( \ket{\lambda_{\ell+i-1}}  + \ket{\lambda_{\ell+i+1}} \right)$. The effect of $\hat{x}$ on $\ket{\psi}$ is then $1 - o(1)$ close to 
\begin{align*}
 \frac{1}{2} c_{\ell+1}   \cdot \left(2 \ket{\psi} + \frac{1}{\sqrt{r-1}}  \ket{\lambda_{\ell}}  +  \frac{1}{\sqrt{r-1}} \ket{\lambda_{\ell+r}}   \right)
 \approx_{r \to \infty}  c_{\ell+1}  \ket{\psi}\,.
\end{align*}
So $\ket{\psi}$ has $1 -o(1)$ overlap with a state supported on eigenvectors of $\hat{x}$, all with eigenvalue $(1 \pm o(1)) \cdot c_{\ell + 1} \approx (1 \pm o(1)) \cdot \sqrt{\ell(m-\ell)}$. 
\end{proof}

\begin{figure}[t]
    \centering
\scalebox{1.2}{
    \begin{tikzpicture}[x=8cm,y=1cm,
        tick/.style={line width=1.2pt,draw=green!50!black},
        note/.style={font=\large}
    ]
    
    % axis
    \draw[-{Latex[length=3pt]},line width=0.9pt,-] (0,0) -- (1,0);
    \fill[white,draw=black,line width=.9pt] (0,0) circle(2.2pt);
    \fill[white,draw=black,line width=.9pt] (1,0) circle(2.2pt);
    
    % labels
    \node[below=6pt] at (0,0) {\scriptsize 0};
    \node[below=6pt] at (1,0) {\scriptsize $m$};
    \node[below=15pt] at (0.2,0) {\scriptsize distance from $v$};
    \node[below=4pt] at (0.5,0) {\scriptsize $m/2$};
    
    % scattered codewords
    \foreach \x in {0.28,0.32,0.35,0.38,0.4,0.42,0.43,0.44,0.45,0.46,0.47,0.48,0.49,0.50,0.51,0.52,
                    0.53,0.54,0.55,0.56,0.57,0.58,0.60,0.62,0.65,0.68,0.72}
      \draw[tick] (\x,-0.11) -- (\x,0.11);

    % beyond decoding threshold
    \foreach \x in {0.14,0.18,0.82,0.95}
      \draw[tick,draw=black] (\x,-0.11) -- (\x,0.11);
      
    % hatched box "found with DQI"
    \draw[draw=blue!70!red,pattern=north east lines,pattern color=blue!70!red, line width=1pt,rounded corners=1pt]
          (0.37,0.2) rectangle (0.63,-0.2);
    \node[note, text=blue!70!red,anchor=south] at (0.50,0.35) {codewords found by DQI};
    
    % % "codewords" with arrows
    % \node[note, text=green!50!black] (cw) at (0.50,1.55) {codewords};
    % \draw[-{Latex[length=5pt]},green!50!black,line width=1pt] (cw.west) .. controls +(225:0.4) and +(90:0.3) .. (0.32,0.3);
    % \draw[-{Latex[length=5pt]},green!50!black,line width=1pt] (cw.east) .. controls +(-45:0.4) and +(90:0.3) .. (0.66,0.3);
    
    \end{tikzpicture}%
}
    
    \caption{\small By \Cref{claim:concentrated_qho_state}, DQI can find codewords at essentially any distance in $[\frac{m}{2} - \sqrt{\ell(m-\ell)}, \frac{m}{2} + \sqrt{\ell(m-\ell)}]$, where $\ell$ is the maximum number of efficiently decodable errors in the dual code. Using classical linear programming techniques (i.e. \Cref{sec:codewords_concentric_hamming_balls}), we can predict the existence of codewords in this region all the way up to the decoding threshold $d^\perp/2$; these codewords are displayed in green. We do not know how to find these codewords classically.} 
    \label{fig:codewords_with_dqi}
\end{figure}

It is now clear how the harmonic oscillator states \textit{themselves} can encode information. In DQI, for example, they encode good approximate solutions to the satisfiability problem $Bx = \vec{1}$.
Even though the obfuscated harmonic oscillator can be analyzed classically, preparing states of this oscillator (which we then sample from) seems to only be possible quantumly.
In a quantum computer, this can be done by coherently error-correcting in the Fourier basis. It is not obvious how to dequantize this procedure.

Our discussion thus far is restricted to $v = \vec{1}$ and binary codes. Our claims generalize to arbitrary $v$, and almost certainly to arbitrary finite fields $\mathbb{F}_p$.

\subsection{Beyond DQI}
Here we have shown how DQI can be formulated as maximizing a parameter (\textit{position}) of the quantum harmonic oscillator, when constrained by another parameter (\textit{energy}) of the oscillator. This viewpoint suggests the possibility of systematically generating DQI-like algorithms, starting from a physics perspective. 

For example, we may use a higher-dimensional oscillator. In two dimensions, the eigenfunctions are no longer Hermite polynomials, but Laguerre polynomials. One may ask how to combine a finite set of Laguerre polynomials in a way that maximizes the radial coordinate, subject to an energy constraint.
Future work should ``reverse-engineer'' this form of DQI by discretizing the physical system, and understanding what the associated optimization task corresponds to. This naturally generalizes to higher dimensions. Other generalizations starting from physics principles include:
\begin{itemize}
    \item Similarly, we may instead use a different potential in the oscillator, like a quartic potential or Lennard-Jones for the radial coordinates.
    \item We may try to maximize another quantity in the oscillator, such as \emph{momentum}.
    \item We may use another, perhaps nonlinear, obfuscating map to store information in the oscillator.
    \item We may use an oscillator with different underlying algebra. DQI over $\mathbb{F}_p$ uses $su(2)$.
\end{itemize}
In each setting, it is worthwhile to find the analogous optimization problem and quantum algorithm. If DQI offers quantum advantage, it is plausible that quantum advantage exists here as well.

%\pagebreak

\section{Discussion}
\label{sec:discussion}
\subsection{Related work}
\paragraph{$\BQP$ and coding theory} 
One of the earliest $\BQP$-complete problems is an extension of the \textit{weight enumerator polynomial} from classical coding theory~\cite{Knill2001}.
The weight enumerator polynomial features prominently in the MacWilliams identity relating the distribution of codewords in $C$ and in $C^\perp$.
This $\BQP$-complete problem can be used to evaluate the partition function of some physical models~\cite{Geraci2008,geraci2010classical}.
Evaluating the weight enumerator polynomial itself is $\PP$-complete~\cite{vyalyi2003hardness}.
One can generalize Shor's algorithm to quantumly estimate some Gauss sums over finite fields~\cite{van2002efficient}. There are several \textit{quantum} variants of the MacWilliams identity used in the study of quantum codes; see for example~\cite{shor1996quantummacwilliamsidentities,huber2018bounds,burchards2025continuousvariablequantummacwilliamsidentities}.

\paragraph{DQI and Regev's reduction}
The ideas behind DQI arguably start with Regev's quantum reduction from the Shortest Vector Problem (SVP) to a problem known as \textit{Learning With Errors} (LWE)~\cite{aharonov2003latticeproblemquantumnp,regev2009lattices}.
The central idea in the reduction is that decoding noisy codewords allows one to find short codewords in the \textit{dual} code (which can be viewed as a lattice).
This technique has been refined to error-correct in \textit{superposition} (e.g. \cite{Chen2022,DebrisAlazard2024}), and even used to construct an $\NP$-search problem that separates $\BQP$ from $\BPP$ relative to a random oracle~\cite{yamakawa2024verifiable}.
Regev's reduction has been partially dequantized, for example when there is a \textit{unique} shortest vector~\cite{peikert2009public,lyubashevsky2009bounded,brakerski2013classical}; see also the survey~\cite{bennett2023complexity}. However, the full reduction has resisted attempts at dequantization.
The DQI algorithm~\cite{dqi} can be viewed as implementing a worst-case variant of the reduction proposed in~\cite{Chailloux:2024dll}.
DQI has been analyzed in some more realistic scenarios, for example with noise~\cite{bu2025decodedquantuminterferometrynoise}, with imperfect decoding~\cite{ct_softdecoders}, and with optimized gate count~\cite{khattar2025verifiable}. For several special cases, DQI is known not to offer quantum advantage~\cite{kothari2025no,parekh2025no,anschuetz2025decoded}, and was generalized in~\cite{schmidhuber2025hamiltonian,sahai_yz}.

\paragraph{Linear programming bounds for coding theory}
Two linear programming bounds comparing the rate and distance of a code were given in~\cite{mceliece1977new}. These arguments rely on a duality argument implicit in the MacWilliams identity, which can be generalized using Delsarte's theory of \textit{association schemes}~\cite{delsarte1973algebraic}; see also the survey~\cite{delsarte2002association}.
The original linear programming bounds have been re-proven using graph theory~\cite{friedman2005generalized}, harmonic analysis~\cite{navon2009linear}, and geometric techniques~\cite{samorodnitsky2023one}.
The linear programs were generalized in~\cite{coregliano2021complete,loyfer2022new,chailloux2024newsolutionsdelsartesdual}.
This duality argument can also be used to bound the rate of a code by the distance of its \textit{dual} code~\cite{tietavainen,sole1995packing}, which we use in this work.
It seems that the weight distribution of linear codes is intimately related to the Kravchuk polynomials~\cite{krawtchouk1929generalisation}; see for example~\cite{samorodnitsky2024weight}.

\paragraph{Quantum harmonic oscillator and Hermite transform}
The eigenstates of a quantum harmonic oscillator on the real line are Hermite polynomials.
It was recently shown that a quantum computer can sample a function by its Hermite coefficients~\cite{iyer2025quantum,jain2025efficient}, leading to polynomial quantum query speedups.
Representing a harmonic oscillator on a digital quantum computer requires a discretization, e.g. \cite{atakishiev1990difference,santhanam2009discrete, somma2016quantumsimulationsdimensionalquantum}.
One choice~\cite{hakioglu2000canonical,chauleur2024discrete} uses the Kravchuk polynomials~\cite{krawtchouk1929generalisation},
and induces a natural physical interpretation~\cite{savva2014integrablemodelsquantummedia,feinsilver2016krawtchoukgriffithssystemsimatrix,feinsilver2016krawtchoukgriffithssystemsiibernoulli}. \cite{dunkl1976krawtchouk} studied the group structure of the Kravchuk polynomials, and their relationship to the Fourier transform was studied in \cite{atakishiyev1997fractional,atakishiyev1999continuous}. Encoding \textit{quantum} information in an oscillator is the core idea of some constructions of \textit{quantum} codes; see for example~\cite{Gottesman_2001,albert2020robust}.

\subsection{Outlook}
\label{sect:outlook}
Outside of the oracle setting, it is hard to justify that a quantum algorithm solves a problem faster than any classical algorithm. 
For example, consider Shor's algorithm: the best evidence that discrete log is not in $\BPP$ is the widespread use of RSA.
Nonetheless, we wish to understand the principles behind algorithms that seemingly offer quantum advantage.

Towards that end, in this work we gave evidence that \textit{simulating} the DQI algorithm is classically hard, despite its lack of a quantum-supremacy-style argument. 
We have shown how DQI hides the subset of codewords at the semicircle law distance; just as in Shor's algorithm, this is an exponentially large set, but an exponentially small fraction of outcomes.
We suggest that DQI and Shor's algorithm have a stronger connection than is currently known. Perhaps there is an interesting quantum circuit framework that unifies these algorithms while also contained in $\BPP^{\NP}$. 

We do not claim that the \textit{optimization} problem DQI solves is classically hard. 
For example, it is not even clear if the solutions output by DQI are locally optimal.\footnote{It is possible that solutions output by DQI could be improved with a \emph{local update}: if changing a bit improves the solution, change this bit with some probability. For example, it is possible to locally improve on the Goemans-Williamson algorithm for Maximum Cut when the input is a random graph~\cite{hastings2021classicalalgorithmbeatsfrac12frac2pifrac1sqrtd}.} 
However, if a classical algorithm can outperform DQI, we suggest that it cannot \textit{directly} simulate DQI. This is in contrast with some proposals for quantum machine learning, which were dequantized by classically simulating the algorithm~\cite{tang_dequantization,dequantizing_via_sampling_framework}.

We also note the similarity to another quantum algorithm for approximately solving optimization problems, namely the QAOA~\cite{farhi2014}. Although the QAOA can create distributions that are impossible to classically sample from (assuming the $\PH$ does not collapse)~\cite{farhi2016quantum}, we still do not know whether it offers quantum advantage \textit{for optimization problems}. 
Nonetheless, the QAOA's claims of quantum advantage (e.g.~\cite{farhi2015quantum,basso2022,farhi2025lower}) have spurred many improvements in classical optimization algorithms (e.g.~\cite{barak2015beating,hastings2019classical,chen2023local}). 
We are hopeful that DQI can present a similar challenge in classical algorithms research. 

We conclude with some direct questions about quantum advantage inspired by our findings:
\begin{enumerate}
    \item We have shown how DQI realizes a known existential result proven using a Fourier transform. Do similar results occur elsewhere in coding theory; if so, can they be made into quantum algorithms? Can one construct DQI-like algorithms that implement more general transforms, such as transforms associated with representations of non-abelian groups~\cite{larocca2025quantum,quantumfire}?
    \item We have shown how a quantum computer can contain information in the states of a harmonic oscillator. In DQI, this information is a specific subset of codewords. Can other information, perhaps beyond coding theory, be stored in a harmonic oscillator? Can one use this framework to learn properties of a known hash function that are hard to learn classically?
    \item We can understand DQI as locating a hidden subset of codewords with some linear-algebraic structure. Is there a \textit{group} we can associate with this structure, perhaps related to the Kravchuk polynomials~\cite{dunkl1976krawtchouk}? In general, what minimal properties are required in a subset $S$ for a quantum computer to successfully locate it?
\end{enumerate}

\section*{Acknowledgements}
Thanks to Eric R. Anschuetz, Noah Shutty, and anonymous reviewers for comments on a draft of this manuscript.
V.H. thanks Sergey Bravyi for discussions. K.M. thanks Tony Metger and Saleh Naghdi for their talks on DQI at Columbia University. 

B.F. and K.M. acknowledge support from AFOSR (FA9550-21-1-0008). This
material is based upon work partially supported by the National Science Foundation under Grant CCF-2044923 (CAREER), by the U.S. Department of Energy, Office of Science, National Quantum Information Science Research Centers (Q-NEXT) and by the DOE QuantISED grant DE-SC0020360.
K.M acknowledges support from the NSF Graduate Research Fellowship under Grant No.\ 2140001.

\printbibliography

\clearpage
\newpage

\appendix

\section{Deferred proofs}
\label{sec:mathfacts}
\begin{definition}
\label{defn:elemsympoly}
The $k^{\text{th}}$ elementary symmetric polynomial $P^{(k)}$ on $m$ inputs is defined as
$$
P^{(k)} \left( \sum_{i=1}^m a_i\right) \defeq P^{(k)}(a_1, \dots, a_m) = \sum_{i_1, \dots, i_k \textnormal{ distinct}} a_{i_1} \times \cdots \times a_{i_k}\,.
$$
\end{definition}
\begin{fact}
\label{fact:count_elementary_sym_polynomials}
Consider a set of $\{a_1, \dots, a_m\}$, where $a_i \in \{\pm 1\}$. Then the degree-$k$ elementary symmetric polynomial $P^{(k)}$ on these inputs is a Kravchuk polynomial; i.e. 
\begin{align*}
P^{(k)}(a_1, \dots, a_m) =  K_k \left(\frac{m-\sum_{i=1}^m a_i}{2}; m \right)\,.
\end{align*}
\end{fact}
\begin{proof}
For notational convenience, define $e_k \defeq P^{(k)}(a_1, \dots, a_m)$.  Let $p_k \defeq \sum_{i=1}^m a_i^k$ be the degree-$k$ ``power sum polynomial'' of $\{a_1, \dots, a_m\}$. Newton's identities relate these two types of polynomials:
$$
k \cdot e_k = \sum_{i=1}^k (-1)^{i-1} \cdot e_{k-i} \cdot p_i\,.
$$
Moreover, if we define the exponential $E(t) = \sum_{k\ge 0} e_k t^k = \prod_{i=1}^m (1 + a_i \cdot t)$, then\footnote{This identity can be found on \href{https://en.wikipedia.org/wiki/Newton's_identities\#Expressing_elementary_symmetric_polynomials_in_terms_of_power_sums}{Wikipedia}; one proof is counting via \href{https://en.wikipedia.org/wiki/Bell_polynomials}{Bell polynomials}.} 
$$
E(t) = \exp( \sum_{k \ge 1} (-1)^{k-1} \frac{p_k}{k} t^k)\,.
$$
When $a_i \in \{\pm 1\}$, the power sum polynomials have only two values: $p_{2d+1} = p_1 = \sum_{i=1}^m a_i$, and $p_{2d} = p_0 = m$. Then
\begin{align*}
 \sum_{k \ge 1} (-1)^{k-1} \frac{p_k}{k} t^k 
&= \left( p_1 \sum_{d \ge 1} \frac{t^{2d-1}}{2d-1}\right) - \left( m \sum_{d \ge 1} \frac{t^{2d}}{2d}\right)
 \\
 &= p_1 \text{ arctanh}(t) - m \left(\frac{-1}{2} \ln(1 - t^2)\right)
 \\
 &= \frac{p_1}{2} \left(\ln(1+t)-\ln(1-t)\right) + \frac{m}{2} \left(\ln(1+t) + \ln(1-t)\right)
 \\
&= \frac{m+p_1}{2} \ln (1+t) + \frac{m-p_1}{2}\ln(1-t)\,.
\end{align*}
So then $\sum_{k \ge 0} e_k t^k = (1+t)^{(m+p_1)/2} \cdot (1-t)^{(m-p_1)/2}$. The coefficients of $t^k$ define the values of $e_k$:
$$
e_k = \sum_{i=0}^k (-1)^{i} {\frac{m-p_1}{2} \choose i} {\frac{m+p_1}{2} \choose k-i}\,.
$$
The proof follows by definition of Kravchuk polynomial (\Cref{defn:krav}).
\end{proof}

\begin{lemma}[{restatement of \Cref{lemma:position_operator_is_bosonic}}]
    Let $\{\ket{\lambda_i}\}$ be the harmonic oscillator states from \Cref{claim:obfuscated_qho} with parity check matrix $B^T: \mathbb{F}_2^m \to \mathbb{F}_2^n$ and distance $d^\perp$. Let $c_\ell \defeq \sqrt{\ell(m-\ell+1)}$. Then the operator
       $$
    \hat{x} \defeq -\frac{1}{2} \sum_{x \in \mathbb{F}_2^n} \sum_{i=1}^m (-1)^{(B x)_i} \ketbra{x}\,,
    $$ 
    satisfies $-2\hat{x} \ket{\lambda_k} = c_k \ket{\lambda_{k-1}} + c_{k+1} \ket{\lambda_{k+1}}$ for all $k < d^\perp/2 - 1$.
\end{lemma}
\begin{proof}
    We rewrite $\hat{x}$ using the Pauli operator $Z = \ket{0}\bra{0} - \ket{1}\bra{1} = \sum_{a} (-1)^a \ket{a}\bra{a}$:
$$
-2\hat{x} 
= \sum_{i=1}^m \sum_x \bigotimes_{j=1}^n \left( (-1)^{B_{ij} x_j} \ket{x_j}\bra{x_j}\right) 
= \sum_{i=1}^m  \prod_{j=1}^n Z_j^{B_{ij}}\,.
$$
The action of $\hat{x}$ on $\ket{\lambda_k}$ is exactly the action of $H_2^{\\otimes n} \hat{x} H_2^{\dagger \otimes n}$ on $H_2^{\otimes n} \ket{\lambda_k}$. The operator is
$$
H^{\otimes n} (-2\hat{x}) H^{\otimes n} = \sum_{i=1}^m \prod_{j=1}^n H_j Z_j^{B_{ij}} H_j = \sum_{i=1}^m (-1)^{v_i} \prod_{j=1}^n X_j^{B_{ij}}\,.
$$
The state has exactly the form from \cref{eqn:oscillator}.
We then apply the operator and get
\begin{align*}
    H^{\otimes n} (-2\hat{x}) \ket{\lambda_k} &= \frac{1}{\sqrt{ {m \choose k}}} \sum_{|z|=k} \sum_{i=1}^m \left( \prod_{j=1}^n X_j^{B_{ij}} \right)\ket{B^T z}
    \\
    &= \frac{1}{\sqrt{ {m \choose k}}} \sum_{|z|=k} \sum_{i=1}^m \ket{B^T z + b_i}
    \\
    &= \frac{1}{\sqrt{ {m \choose k}}} \sum_{|z|=k} \ket{B^T (z + e_i)}\,,
\end{align*}
where $b_i$ is the $i^\text{th}$ row of $B$, and $e_i \in \mathbb{F}_2^m$ is the basis vector for coordinate $i$.

\begin{itemize}
\item We first consider $k = 0$. Then this is equal to $m \cdot \sum_{|y|=1} \ket{B^T y} = \sqrt{m} \cdot H^{\otimes n} \ket{\lambda_1}$.
    \item Now consider $1 \le k < d^{\perp}/2 - 1$. Then $(z+e_i)$ will have Hamming weight either $k-1$ or $k+1$. Every vector of weight $k-1$ will be counted $m-(k-1)$ times; every vector of weight $k+1$ will be counted $k+1$ times. So
\begin{align*}
    H^{\otimes n} (-2\hat{x}) \ket{\lambda_k} &= \frac{1}{\sqrt{ {m \choose k}}} \left( (m-(k-1)) \sum_{|z|=k-1} \ket{B^T z} + (k+1)\sum_{|z|=k+1} \ket{B^T z}\right)\,.
\end{align*}
The first term is $H^{\otimes n} \ket{\lambda_{k-1}}$ up to normalization.  When $k < d^\perp/2 - 1$, the second term is $H^{\otimes n} \ket{\lambda_{k+1}}$ up to normalization. Recall that ${m \choose k} = \frac{m-k+1}{k} {m \choose k-1} = \frac{k+1}{m-k} {m \choose k+1}$. Then
\begin{align*}
H^{\otimes n} (-2\hat{x}) \ket{\lambda_k}  &= \sqrt{\frac{k}{m-k+1}} (m-(k-1))  H^{\otimes n}  \ket{\lambda_{k-1}} + \sqrt{\frac{m-k}{k+1}}(k+1) H^{\otimes n} \ket{\lambda_{k+1}}
\\
&= H^{\otimes n} \left( c_k \ket{\lambda_{k-1}} + c_{k+1} \ket{\lambda_{k+1}} \right)\,,
\end{align*}
where $c_\ell \defeq \sqrt{\ell(m-\ell+1)}$.
\end{itemize}
We combine these cases. Applying $H^{\otimes n}$ to both sides, we get
$
-2\hat{x} \ket{\lambda_k} = c_k \ket{\lambda_{k-1}} + c_{k+1} \ket{\lambda_{k+1}}
$.
\end{proof}
% \clearpage
% \newpage

\clearpage
\newpage
\section{Computing the output probabilities of Shor's algorithm}
\label{sec:simulate_shor}
We briefly explain how Shor's algorithm for discrete logarithm prepares a state of the form 
$$
\sum_x w(x)\ket{x}\,,
$$
where $|w|^2$ is efficiently computable. By \Cref{thm:sampling_nonnegative_weights}, this state can be efficiently sampled given a classical machine with access to an $\NP$ oracle.

In the discrete logarithm problem, we are given an element $x$ of a cyclic group $G = \langle g \rangle$, and we wish to calculate the smallest integer $\alpha$ such that $g^\alpha = x$. We denote $\log_g x \defeq \alpha$. Suppose we know $N$ is the order of $G$; i.e. $G = \{1, g, g^2, \dots, g^{N-1}\}$. In general $N$ is exponential in the system size.

Shor's algorithm (see \cite{childs2017lecture} for a review) prepares the state
\begin{align*}
    \frac{1}{N} \sum_{\delta, \nu \in \mathbb{Z}_N}\omega_N^{\nu \delta} \ket{\nu \log_g x, \nu, g^\delta}\,.
\end{align*}
Suppose we are given the outcome $\ket{a,b,c}$. Then the probability of observing this outcome is nonzero if and only if $a = b \cdot \log_g x$, i.e.
$$
g^{a} = x^b\,.
$$
There are $N^3$ possible outcomes $\ket{a,b,c}$, so the probability of any nonzero outcome is exactly $\frac{1}{N^2}$. Then the squared amplitude $|w|^2$ is efficiently computable:
$$
|w(a,b,c)|^2 = \frac{1}{N^2} \cdot \mathbf{1}_{g^a = x^b}\,.
$$
Note that this value is much larger than the average probability $\frac{1}{N^3}$, but much smaller than any inverse polynomial in system size.

\clearpage
\newpage

\section{Existence of codewords in the semicircle law region}
\label{sec:codewords_concentric_hamming_balls}
We extend the techniques of \Cref{sec:semicirclelaw} to show that for any point $v \in \mathbb{F}_2^m$, codewords exist not only \emph{at} the semicircle law distance $\alpha$ from $v$, but at essentially every distance from $v$ in $[\alpha, m-\alpha]$. This relies on the following fact:
\begin{fact}[{e.g. \cite[Proposition 5.1]{coleman2011krawtchouk}}]
\label{fact:krav_simpleroot}
For any $1 \le k \le m$, the Kravchuk polynomial $K_k(\cdot;m)$ has $k$ distinct roots.
\end{fact}

\Cref{claim:function_to_bound_coveringradius} gives an upper bound on the distance from any vector $v$ to a codeword. But with a minor tweak, it can also show the existence of codewords within a small \textit{range} of distances from $v$.
\begin{claim}
\label{claim:function_to_bound_hammingball_slice}
    Fix any binary linear code $C$ with dual distance $d^\perp$. Consider any $\beta(x) = \sum_{k=0}^{d^\perp - 1} \beta_k K_k(x)$ such that $\beta_0 > 0$ and $\beta(i) \le 0$ for all $i > j_2$ and all $i < j_1$. Then for any target vector $v$, there exists a codeword $c \in C$ at distance $|c-v| \in [j_1,j_2]$.
\end{claim}
\begin{proof}[Proof sketch]
    Exactly follows the proof of \Cref{claim:function_to_bound_coveringradius}, except that $i \in [j_1, j_2]$ in \cref{eqn:bounding_betazero}.
\end{proof}
It remains to find good choice of $\beta(x)$ that satisfies the conditions of \Cref{claim:function_to_bound_hammingball_slice}. We propose:
\begin{claim}
\label{claim:beta_hammingball_slice}
    Fix an even $\ell$, $r \in [\ell - 1]$, and $\epsilon > 0$. Let $\alpha_1 \le  \dots \le \alpha_\ell$ be the roots of $K_\ell(x)$. Then
    $$
    \beta(x) = -(x-\alpha_1)^2 \dots (x-\alpha_{r-1})^2 (x-\alpha_r + \epsilon)(x-\alpha_{r+1}-\epsilon) (x-\alpha_{r+2})^2 \dots (x-\alpha_\ell)^2\,,
    $$
    has $\beta_0 > 0$, and $\beta(i) \le 0$ for all $i > \alpha_r-\epsilon$ and $i < \alpha_{r+1}+\epsilon$.
\end{claim}
\begin{proof}
    By construction, $\beta(x)$ is positive only when $-(x-\alpha_r + \epsilon)(x-\alpha_{r+1}-\epsilon)$ is positive. This only occurs when $x \in [\alpha_r - \epsilon, \alpha_{r + 1} + \epsilon]$.

    Now we prove $\beta_0 > 0$. By \Cref{fact:krav_orthogonal}, when $\epsilon = 0$, 
    $$
    \beta_0 = \sum_{i=0}^m {m \choose i} \beta(i) = \sum_{i=0}^m {m \choose i} K_{\ell}(i) \cdot \frac{K_{\ell}(i)}{(x-\alpha_r)(x-\alpha_{r+1})}\,.
    $$
    Again by \Cref{fact:krav_orthogonal}, this must be zero, since the second term has degree strictly less than $\ell$.

    We take a derivative of $\beta_0$ with respect to $\epsilon$. This value is
    \begin{align*}
        &\sum_{i=0}^m {m \choose i} \left[ (x-\alpha_r + \epsilon) - (x-\alpha_{r+1} - \epsilon) \right] \left( \prod_{s=1}^{r-1} (i-\alpha_s)^2 \right) \left( \prod_{s=r+2}^{\ell} (i-\alpha_s)^2 \right)
        \\
        &= \sum_{i=0}^m {m \choose i} \left[ (\alpha_{r+1} - \alpha_r) + 2\epsilon \right] \left( \prod_{s=1}^{r-1} (i-\alpha_s)^2 \right) \left( \prod_{s=r+2}^{\ell} (i-\alpha_s)^2 \right)\,.
    \end{align*}
    By \Cref{fact:krav_simpleroot}, the roots are all distinct, so $\alpha_{r+1} > \alpha_r$. Since $\epsilon \ge 0$, this value is positive, and so $\beta_0 > 0$ for positive $\epsilon$.
\end{proof}
Together, \Cref{claim:function_to_bound_hammingball_slice} and \Cref{claim:beta_hammingball_slice} imply there are codewords at distance between any two roots of the Kravchuk polynomial $K_\ell$, for any $2\ell + 1 < d^\perp$. As $m \to \infty$ but $\ell/m$ is held constant, these roots are in the range $[m/2 - \sqrt{\ell(m-\ell)}, m/2 + \sqrt{\ell(m-\ell)}]$ (e.g. \cite{ismail1998strong,kirshner2021moment}). Since the roots of $K_{\ell}$ and $K_{\ell-1}$ interlace~\cite{szego1939orthogonal,chihara1990zeros}, there are codewords at essentially every distance in the range $[\frac{m}{2}-\sqrt{\frac{d^\perp}{2}(m-\frac{d^\perp}{2})}, \frac{m}{2}+\sqrt{\frac{d^\perp}{2}(m-\frac{d^\perp}{2})}]$.

\end{document}